\documentclass[11pt,a4paper]{article}
\usepackage[margin=1in]{geometry}
\usepackage{lmodern}
\usepackage{amsthm}
\usepackage{cite}
\usepackage{caption}
\theoremstyle{plain}
\newtheorem{theorem}{Theorem}
\newtheorem{lemma}[theorem]{Lemma}
\newtheorem{corollary}[theorem]{Corollary}
\theoremstyle{definition}
\newtheorem{definition}[theorem]{Definition}
\newtheorem{construction}[theorem]{Construction}
\newtheorem{assumption}[theorem]{External building block}
\newcommand{\doi}[1]{\href{https://doi.org/#1}{\nolinkurl{doi:#1}}}
\usepackage[T1]{fontenc}
\usepackage{amsmath,amssymb,mathtools,bm}
\usepackage{microtype,booktabs,array,tabularx,longtable}
\usepackage{enumitem,needspace,listings}
\usepackage{graphicx,tikz,float}
\usetikzlibrary{arrows.meta,positioning,calc,fit,backgrounds}
\tikzset{
 bbjarr/.style={-{Latex[length=1.8mm]},line width=.45pt},
 bbjcmp/.style={dashed,line width=.5pt},
 bbjbox/.style={draw,line width=.45pt,fill=white,align=center,inner sep=5pt,font=\small},
 bbjstate/.style={align=center,inner sep=2pt,font=\small},
 bbjlab/.style={align=center,fill=white,inner sep=2pt,font=\footnotesize}
}
\usepackage{hyperref}
\hypersetup{hidelinks,bookmarksdepth=3,pdftitle={Rate 1/5 Non-Malleable Codes against Entangled Split-State Tampering},pdfauthor={Naresh Goud Boddu},pdfsubject={Classical messages and classical codewords; information-theoretic two-split non-malleability}}
\usepackage[nameinlink,noabbrev]{cleveref}
\crefname{theorem}{Theorem}{Theorems}
\crefname{lemma}{Lemma}{Lemmas}
\crefname{corollary}{Corollary}{Corollaries}
\crefname{definition}{Definition}{Definitions}
\crefname{section}{Section}{Sections}
\crefname{table}{Table}{Tables}
\crefname{construction}{Construction}{Constructions}
\crefname{assumption}{External building block}{External building blocks}
\newcommand{\bits}{\{0,1\}}
\newcommand{\F}{\mathbb F}
\newcommand{\E}{\mathbb E}
\newcommand{\Tr}{\operatorname{Tr}}
\newcommand{\Id}{\mathsf{id}}
\newcommand{\Enc}{\mathsf{Enc}}
\newcommand{\Dec}{\mathsf{Dec}}
\newcommand{\nmext}{\mathsf{2nmExt}}
\newcommand{\Copy}{\mathsf{Copy}}
\newcommand{\Sym}{\mathsf{Sym}}
\newcommand{\LeftShare}{L_{\mathrm{sh}}}
\newcommand{\RightShare}{R_{\mathrm{sh}}}
\newcommand{\sflag}{s_{\mathrm{flag}}}
\newcommand{\same}{\mathsf{same}}
\newcommand{\pref}{\operatorname{pref}_{\ell}}
\newcommand{\nz}{\operatorname{nz}}
\newcommand{\ind}{\mathbf 1}
\newcommand{\supp}{\operatorname{supp}}
\newcommand{\ket}[1]{\lvert#1\rangle}
\newcommand{\bra}[1]{\langle#1\rvert}
\newcommand{\normone}[1]{\left\lVert#1\right\rVert_1}

\newcommand{\epsp}{\varepsilon_{\rm pair}}
\newcommand{\epsn}{\varepsilon_{\rm NMC}}

\newcommand{\cK}{\mathcal K}
\newcommand{\cA}{\mathcal A}

\newcommand{\heading}[1]{\par\smallskip\noindent\textbf{#1}\enspace}
\newcolumntype{L}[1]{>{\raggedright\arraybackslash}p{#1}}
\allowdisplaybreaks[2]
\setlist[enumerate]{itemsep=2pt,topsep=4pt}
\setlist[itemize]{itemsep=2pt,topsep=4pt}
\begin{document}
\title{Rate 1/5 Non-Malleable Codes\\against Entangled Split-State Tampering}
\author{Naresh Goud Boddu\thanks{Email: \href{mailto:ngboddu93@nus.edu.sg}{\nolinkurl{ngboddu93@nus.edu.sg}}}\\[3pt]{\normalsize Department of Computer Science, National University of Singapore}}
\date{}
\maketitle

\begin{abstract}
We construct efficient information-theoretic non-malleable codes for 
classical messages that are secure against two noncommunicating local quantum tampering operations with arbitrary
pre-shared entanglement. For every sufficiently small fixed $\xi>0$ and
all sufficiently large first-share lengths $n$, the codes have rate at least
$1/5-\xi$, perfect correctness, and error $2^{-n^{\Omega(1)}}$. Security holds for every message, with a single
message-independent simulator for each attack. This resolves the
constant-rate question for worst-case classical messages in the entangled
two-split-state model.

Our construction builds on the permutation-based two-split construction
of Batra, Boddu, and Jain, which achieves rate approaching $1/5$ for
uniformly random messages. We retain their architecture but replace the
uniform message input to the permutation with a prescribed message
concatenated with fresh uniform padding. Our main contribution is a
worst-case security reduction for this modification.

\end{abstract}

\section{Introduction}
\label{sec:intro}

\subsection{Non-malleability under local access}
Learning a secret is not the only useful outcome of an attack. An adversary
may instead modify a ciphertext so that it decrypts to a value with an
adversarially chosen relationship to the original plaintext, without
learning either plaintext. For example, replacing a
one-time-pad ciphertext $c=m\oplus k$ by $c\oplus\Delta$ replaces its
plaintext by $m\oplus\Delta$. Perfect secrecy does not prevent this attack.
Likewise, hiding a bid in a commitment is insufficient if another participant
can commit to a related bid without knowing the original amount. Dolev, Dwork, and
Naor~\cite{DDN91} introduced non-malleable cryptography to formalize this
separation between hiding an object and controlling its relationship with
another object, studying encryption, commitments, and proofs of knowledge.
The goal is not to exclude accidental relations, but to rule out dependence
unavailable to an independent simulation. Bellare, Desai,
Pointcheval, and Rogaway~\cite{BDPR98} subsequently clarified the relationships
between non-malleability and indistinguishability under different
chosen-plaintext and chosen-ciphertext attack models. Cramer and
Shoup~\cite{CS98} gave a practical public-key encryption construction with
adaptive chosen-ciphertext security in the standard model, illustrating how
the concern with related ciphertexts leads to stronger security than hiding
alone.

\heading{From malleable ciphertexts to tamperable stored state.}
The same concern arises when a cryptographic device's \emph{internal state},
rather than a protocol message, can be modified. Security under one secret
key does not automatically describe executions under adversarially related
keys. Dziembowski, Pietrzak, and Wichs~\cite{DPW10} introduced non-malleable
codes (NMCs) to address the representation of such stored state. NMCs ensure that  any dishonest party tampering the internal state (corresponding to a secret key) leads to decrypted key being either the original secret key or something unrelated to it. Unlike a keyed encryption or authentication
primitive, an NMC has public encoding and decoding algorithms and needs no
separately protected secret key. Security instead requires a restriction on
codeword access.

The relaxation from error correction or detection is essential. If an attacker
can overwrite a codeword with a valid encoding of a fixed value $m_0$, the
decoder cannot recover the previous message or reject that encoding without
violating correctness for $m_0$. An NMC permits this independent replacement.
For each allowed attack, one message-independent distribution chooses either
$\same$, interpreted as the original message, or a replacement value, possibly
failure. The same distribution must work for every input message. The original
work~\cite{DPW10} established feasibility for bitwise-independent tampering
and sufficiently small families of functions, initiating the study of which
structural restrictions admit efficient non-malleable codes.

One of the best-studied restrictions is the split-state model.\footnote{An unrestricted joint procedure can decode, change the message, and re-encode;
no public coding scheme prevents that attack.} In its two-state form,
the codeword consists of two strings $\LeftShare$ and $\RightShare$, and tampering acts
as $(f(\LeftShare),g(\RightShare))$. Each tampering procedure may change its whole
component, but the procedures cannot exchange information during tampering. Thus the restriction concerns locality of access, not the number of changed bits. It models separately accessible distributed state,
without assuming that every physical fault respects this separation. The survey of Aggarwal, Ball, and Obremski~\cite{ABOsurvey} develops this motivation and the resulting cryptographic applications. Two components are
the smallest nontrivial split under unrestricted local computation. Constant
rate gives constant-factor storage overhead for long protected state.

Liu and Lysyanskaya~\cite{LL12} provide security against split-state
leakage and tampering in the common-reference-string model. Dziembowski, Kazana, and Obremski~\cite{DKO13} give an
information-theoretic two-split construction for one-bit messages;
Aggarwal, Dodis, and Lovett~\cite{ADL14} obtain multi-bit encodings through
additive combinatorics. Cheraghchi and Guruswami~\cite{CGcoding14} connect
split-state non-malleable coding to seedless non-malleable extraction, providing a route
to construct non-malleable codes in the split-state model. Beyond storage, Goyal,
Pandey, and Richelson~\cite{GPR16} use split-state NMCs in non-malleable
commitments, while Goyal and Kumar~\cite{GK18} develop non-malleable secret
sharing beyond the two-out-of-two setting. Such applications motivate the
primitive but can require augmentation, leakage resilience, or other
properties beyond the non-malleability that is proved here. Sections~\ref{sec:history-codes} trace the broader
history and the connections to secret sharing, commitments, and
round-efficient multiparty computation.

\heading{Why allow shared entanglement before tampering?}
Information-theoretic classical split-state security already allows unbounded computation, so its
quantum extension is not merely protection against faster algorithms.
Shared classical coins produce mixtures of product response distributions.
Local measurements of pre-shared entanglement can produce correlations outside
that class without communicating. Aggarwal, Boddu, and Jain~\cite{ABJ24}
formalize this stronger split-state adversary; Batra, Boddu, and Jain
(BBJ)~\cite{BBJ23} develop the
high-rate non-malleable randomness encoder (NMRE) primitive and use it to achieve non-malleability for uniform messages in the split-state model on which we
build. Our messages and codewords remain classical; the proof must
preserve entangled correlations while obtaining one simulator for every message.

\subsection{Our result: worst-case security at rate approaching one-fifth}
\label{sec:main-result-overview}
Prior quantum-secure two-split NMCs protected every classical message at vanishing inverse-polynomial rate~\cite{ABJ24}, or achieved rate approaching $1/5$ only for a uniform
message~\cite[Section 7]{BBJ23}. Constant rate with negligible error for
\emph{every fixed message} against entangled two-split tampering remained
open. We resolve this combination of requirements.

Here and below, $n$ denotes the length of the first share.
\begin{theorem}[Main result, informal]\label[theorem]{thm:informal}
For every sufficiently small fixed $\xi>0$, there is an efficient classical
two-split-state non-malleable code of rate at least $1/5-\xi$ for all
sufficiently large $n$, with perfect correctness and error
$2^{-n^{\Omega(1)}}$, secure against every pair of noncommunicating local
quantum tampering operations with arbitrary pre-shared entanglement.
\end{theorem}

The construction modifies the permutation-based two-split architecture of
BBJ~\cite{BBJ23}. Let $m\in\bits^\ell$. Independently sample
$X\in\bits^n$, $Y\in\bits^d$, and padding $R\in\bits^r$. Set
$s=\ell+r$ and use $Z=\nmext(X,Y)\in\bits^{2s}$ to specify an affine
permutation $\pi_Z$ on $s$-bit strings:
\begin{equation}
 C=\pi_Z(m\|R),\qquad \Enc(m)=\bigl(X,(Y,C)\bigr).
 \label{eq:introcode}
\end{equation}
Write $X'$ and $(Y',C')$ for the tampered shares. The decoder recomputes
$Z'=\nmext(X',Y')$, inverts the permutation, and outputs the first
$\ell$ bits as $M'$. The new ingredient is the \emph{worst-case reduction}, not the use of a permutation key derived from a
non-malleable randomness encoder.

For a quantum secure non-malleable randomness encoder $\nmext$ with error $\eta$, our theorem gives
$O(\sqrt\eta+2^{-(r-d)/5})$ error when $r>d$, where $r$ is the length of the padding and $d$ is the length of source $Y$.  Instantiating $\nmext$ with BBJ's construction and setting $r=2d$, the asymptotic rate of our construction (for a sufficiently small fixed $\delta>0$) is
\begin{equation}
 \frac{1-10\delta}{5+2\delta}.
 \label{eq:intro-rate}
\end{equation}
The above rate counts the sum of the two unequal share lengths. A balanced-share corollary in Section~\ref{sec:balanced} approaches $1/8$ with the same security error.

\subsection{Closest comparison: what is inherited from BBJ, and what changes}
\label{sec:bbj-comparison}
BBJ's two-split construction already derives a permutation key from its non-malleable 
randomness encoder, stores the sources as $X,Y$, and places the permuted
uniform message along with $Y$. They establish rate approaching $1/5$ for uniform classical messages
against the same entangled local adversaries~\cite[Section 7]{BBJ23}. Table~\ref{tab:construction-comparison} identifies the change in
message handling and analysis.

\begin{table}[H]
\caption{The closest construction-level comparison.}
\label{tab:construction-comparison}
\centering\small\setlength{\tabcolsep}{3pt}
\begin{tabularx}{\textwidth}{@{}L{3.1cm}>{\raggedright\arraybackslash}X>{\raggedright\arraybackslash}X@{}}
\toprule
Feature & BBJ, two-split code~\cite[Section 7]{BBJ23} & This work\\
\midrule
Architecture & Uses an NMRE to select a permutation & Same architecture\\[4pt] \\
Input to permutation & Uniform message & Fixed message $m$ padded with uniform $R$\\[4pt] \\
Reordered experiment & Uses uniformity of the whole message & First produce $Y'$ from $Y$ (for uniform $C$), then recover $C,C'$ \\[4pt]\\
Message analysis & Joint analysis for a uniformly random message & Mean-one selection with two exact normalizations and a bounded deficit\\[4pt]\\
Conclusion & Uniform-message security; rate near $1/5$ & Worst-case security; rate near $1/5$; error $O(\sqrt\eta+2^{-(r-d)/5})$\\
\bottomrule
\end{tabularx}
\end{table}

\heading{The precise point where uniformity is needed.}
The opening of BBJ's proof of Theorem 7 identifies its original and reordered
experiments using that the whole message is uniform and the keyed map is a
bijection~\cite[Theorem 7, proof; Figures 11 and 13]{BBJ23}. For each fixed
permutation key, permuting a uniform string leaves a uniform ciphertext. A fixed
message does not allow that substitution. In our padded code it instead
selects a key-dependent set of $2^r$ ciphertexts from $2^s$ possibilities.
The real encoding is the uniform-ciphertext experiment conditioned on one
message, an event of probability $2^{-\ell}$. Applying an average-case error
bound after this conditioning can multiply the inner error by $2^\ell$. This is a limitation of that route to a stronger theorem, not a gap in BBJ's stated uniform-message result.

Padding creates enough residual randomness for a different comparison. We
apply BBJ before recovering the long ciphertext, then analyze the message
selection by its deficit on average rather than its maximum loss due to conditioning. The continuation
identity preserves residual quantum state, and the permutation second order moment turns
$r-d$ into a positive security surplus. Thus padding alone is not a proof: it is padding together with this new conversion that yields the worst-case statement.

\subsection{Proof overview I: classical adversaries}
\label{sec:classical-overview}
We first explain the classical argument, proved fully in
Appendix~\ref{app:classical}. Let the message length be $\ell$ and the padding randomness length be $r$. Write $s=\ell+r$, $N=2^s$, $q=2^\ell$,
and $T=2^r$, so $N=qT$; $f_z$ is the inverse permutation followed by taking the prefix.
Consider deterministic tampering $X'=f(X)$ and
$(Y',C')=(g_1(Y,C),g_2(Y,C))$. The proof for shared classical coins averages
the resulting conditional output distributions over their source-independent tapes.

\heading{The obstacle is a rare message event.}
In the message-free hybrid, $X,Y,C$ are independently uniform. For every permutation key $Z = \nmext(X,Y)$, 
exactly $T$ ciphertexts decode to $m$, so
$\Pr[f_Z(C)=m\mid X,Y]=1/q$. Conditioning this hybrid on that event $f_Z(C)=m$ reproduces the real encoding of $m$, with its original source marginal
unchanged. Conditioning an \emph{approximate} hybrid statement can nevertheless cost a
factor $q$. Padding leaves this probability unchanged but enlarges the size of $C$. We need a comparison that exploits its extra randomness.

\heading{Produce $Y'$ before recovering the ciphertexts $C,C'$.}
Let $A_{y,y'}=\{c:g_1(y,c)=y'\}$ and $h_{y,y'}=|A_{y,y'}|$.
The hybrid's conditional distributions are
\[
 P^{\rm hyb}_{Y'\mid Y}(y'\mid y)=\frac{h_{y,y'}}N,\qquad
 P^{\rm hyb}_{C\mid Y,Y'}(c\mid y,y')
       =\frac{\ind[c\in A_{y,y'}]}{h_{y,y'}}
\]
when $h_{y,y'}>0$. Their product equals
$N^{-1}\ind[g_1(y,c)=y']$, the original joint conditional probability
$\Pr^{\rm hyb}[Y'=y',C=c\mid Y=y]$. Consequently the same joint distribution is
generated by sampling $Y'$ from $P^{\rm hyb}_{Y'\mid Y}$, then $C$ from $P^{\rm hyb}_{C\mid Y,Y'}$, and setting $C'=g_2(Y,C)$. The first stage $Y \rightarrow Y'$ is
a valid attack on $Y$ alone; it has no message-correlated ciphertext input. 

At this stage, we invoke the security guarantee of BBJ's non-malleable
randomness encoder, $\nmext$, retaining the register
$\Lambda=(Y,Y',S,\overline Z')$, where $S$ indicates
$(X',Y')=(X,Y)$ and $\overline Z'$ masks the tampered key on $S=1$. These registers are computable before recovering $C$. Let $\mu$ be the distribution
of $(Z,\Lambda)$. We call $(Z,\Lambda)$ the \emph{coarse register};
continuation map appends $C,C'$ to form the \emph{refined register}. BBJ's
security guarantee compares the original key with an independent ideal
affine key on $\cK=\F_N^\times\times\F_N$, giving
\[
 \sigma=U_{\cK}\times\mu_\Lambda,\qquad
 \Delta(\mu,\sigma)\le\alpha:=\eta+N^{-1}.
\]
The $N^{-1}$ cost is to restrict a fresh uniform key to nonzero slope.
Continuation map depends only on $y,y'$.

\heading{Use a bounded deficit instead of the largest likelihood ratio.}
For a fixed message and $h_{y,y'}>0$, selecting the ciphertexts $c\in A_{y,y'}$ satisfying $f_z(c)=m$ and multiplying by $q$ gives the
coarse register $(z,\lambda)$ and the total output scaling factor as a result of conditioning
\[
 w_m(z,\lambda)=\frac q{h_{y,y'}}
                   \sum_{c\in A_{y,y'}}\ind[f_z(c)=m].
\]
Both $\E_\mu w_m$ and $\E_\sigma w_m$ equal one. The actual equality ($\E_\mu w_m=1$) uses the fact that for every fixed $z$, $\vert\{c:f_z(c)=m\}\vert=T$ for every $m$. The ideal equality ($\E_\sigma w_m=1$) uses $\Pr_{z\leftarrow U_{\cK}}[f_z(c)=m]=1/q$ for each fixed $c$. For each fixed pair $(z,\lambda)$, $w_m(z,\lambda)$ is the total
scaling factor assigned by message selection.

Although $w_m$ can reach $q$~(for example, this can happen if tampering is such that $h_{y,y'}=1$ for some $(y,y')$), its \emph{deficit} $(1-w_m)_+$ on average lies in
$[0,1]$. For any distribution $P$ with $\E_Pw=1$, the scaled distribution $Pw$ obeys
\[
 \Delta(Pw,P)=\tfrac12\E_P|w-1|=\E_P(1-w)_+.
\]
Excess mass above scaling factor one must equal the missing mass below it.
Thus $\Delta(\mu,\sigma)\le\alpha$ transfers deficits at cost at most
$\alpha$, rather than $q\alpha$. If $d_\sigma=\E_\sigma(1-w_m)_+$,
comparison along $\mu w_m,\mu,\sigma,\sigma w_m$ costs at most
$2\alpha+2d_\sigma$. The remaining step uses the same conditional output distribution on both scaled distributions, so total variation cannot increase. This is the scalar mean-one lemma, not ordinary data processing of the unnormalized message-selection map.

\heading{Padding makes the average deficit small.}
Under the product distribution $ \sigma=U_{\cK}\times\mu_\Lambda$, the conditional second moment satisfies
$\E_\sigma[(w_m-1)^2\mid\Lambda]\le q/h_{y,y'}$, and the conditional mean
is one. Thus
$d_\sigma\le\tfrac12\E_{(y,y')\leftarrow P^{\rm hyb}_{YY'}}\sqrt{q/h_{y,y'}}$.
Small posterior sets can have poor conditional bounds, but their hybrid
probabilities are $P^{\rm hyb}_{YY'}(y,y')=2^{-d}h_{y,y'}/N$.
Since the sets for each $y$ partition $N$ points into at most $2^d$ parts,
Cauchy--Schwarz gives
\[
 d_\sigma\le\frac{\sqrt q}{2N}
       \sum_y2^{-d}\sum_{y':h_{y,y'}>0}\sqrt{h_{y,y'}}
 \le\frac12\sqrt{\frac{q2^d}{N}}
 =\frac12\,2^{-(r-d)/2}.
\]
The surplus is $s-\ell-d=r-d$.

\heading{Convert the ideal distribution into one simulator.}
Compute the output symbol before discarding the original key. On $S=0$
use the actual tampered-key decoder output $f_{Z'}(C')$; on $S=1,C'=C$ output $\same$; on
$S=1,C'\ne C$ compute $f_Z(C')$. Copying this symbol reproduces the
real decoder on every execution. Only after the ideal-key comparison do we
use that distinct ciphertexts have uniformly distributed distinct inverse images. Conditional on the first prefix being $m$, the second prefix has
mass $(T-1)/(N-1)$ at $m$ and $T/(N-1)$ elsewhere. Replacing this distribution by uniform costs $\epsp=(q-1)/(q(N-1))$.

The ideal message selection leaves each refined register's hybrid scaling factor
unchanged. Therefore a simulator that runs the message-free hybrid and
uses fresh uniform output in the last case is fixed independently of $m$.
The full classical bound is
\[
 \epsn^{\rm cl}\le2(\eta+2^{-s})+2^{-(r-d)/2}+\epsp.
\]
This uses no quantum inequality. 

\subsection{Proof overview II: extending the comparison to entangled adversaries}
\label{sec:quantum-overview}
The encoder and simulator cases stay the same, but the continuation map and
scaling factor arguments change. An entangled attack leaves a state in the right auxiliary register correlated
with the other component; selecting an outcome can change that conditional state. Reproducing scalar probabilities is therefore
insufficient: we must reconstruct the complete conditional quantum state.

\heading{Reconstruct the conditional state, not only the event probability.}
The right attack sends the classical inputs $(Y,C)$ to $(Y',C')$ and
updates the state from $E_2$ to $E'_2$. Its completely positive (CP) map
$\mathcal V^{y,c}_{y',c'}$ gives the probability-scaled output state for
input values $(y,c)$ and output values $(y',c')$. Define its output probability operator and its uniform-ciphertext average by
\[
 M_{y'}^{y,c}=\sum_{c'}(\mathcal V^{y,c}_{y',c'})^*(I),\qquad
 G_{y'}^y=\frac1N\sum_cM_{y'}^{y,c}.
\]
Thus $\Tr G_{y'}^y\theta$ is the probability of $Y'=y'$ on input state $\theta_{E_2}$
in the hybrid. With smoothing\footnote{Smoothing is a known quantum information theoretic technique to ensure invertibility of positive semidefinite operators in the analysis.} probability $\zeta$, choose $Y'$
uniformly, keep $C'=C$, and retain the input state in $E_2$ instead of
running the original right operation. Use the same notation for the
resulting smoothed maps and operators. This changes the final output by
at most $\zeta$ and ensures $G_{y'}^y\succeq\zeta2^{-d}I$ on every input
direction.

The coarse CP map is $\mathcal C_{y'}^y(\theta)=\sqrt{G_{y'}^y}\theta\sqrt{G_{y'}^y}$.
It produces the correct event probability, but generally not the original
conditional state in $E_2$. With $G=G_{y'}^y$, the continuation is
\[
 \mathcal R^{y,y'}_{c,c'}(\theta)
   =N^{-1}\mathcal V^{y,c}_{y',c'}(G^{-1/2}\theta G^{-1/2}),\qquad
 \mathcal R^{y,y'}_{c,c'}\circ\mathcal C_{y'}^y
   =N^{-1}\mathcal V^{y,c}_{y',c'}.
\]
Adjacent inverse square roots cancel in order. For each $y,y'$, continuation
sums to a channel on all inputs, and the composition is an identity on all
matrices. It therefore survives tensoring with the identity on the left
attacker's register $E_1$: this replaces the scalar posterior identity.

\heading{Apply BBJ's security guarantee before recovering $C,C'$.}
The coarse procedure acts only on registers $Y,E_2$; its definition
has averaged over $C$ and has no original-key input. BBJ can therefore be
applied once, yielding an actual state $\mu_{Z\Lambda E_2}$ and
$\sigma=U_{\cK}\otimes\mu_{\Lambda E_2}$ within
$\alpha=\eta+N^{-1}$. The allowed register
$\Lambda=(Y,Y',S,\overline Z')$ and register $E_2$ retain their actual joint
state in $\sigma$.
The continuation's probability operators for $C=c$ satisfy
\[
 Q_c=\sum_{c'}(\mathcal R^{y,y'}_{c,c'})^*(I),\qquad
 \sum_cQ_c=I,\qquad Q_c\preceq\kappa I,
 \quad\kappa=\frac{2^{d-s}}{\zeta}.
\]
Only the $2^d$ outcomes are charged; continuation produces the long
$C'$ later, and no auxiliary-register dimension cost is incurred.

\heading{Replace scaling factor by a positive Hermitian operator.}
For a fixed message $m$, the map $\Phi_m$ continues the experiment, selects
$f_Z(C)=m$, multiplies by $q$, and computes the three-case symbol while
$Z$ is still available. It is a CP map but not a channel on arbitrary inputs.
Its output-scaling factor operator $H=\Phi_m^*(I)$ has blocks
$H_z^\lambda=q\sum_{c:f_z(c)=m}Q_c^{y,y'}$.
The two normalizations survive exactly: on the actual input ($\mu_{Z\Lambda E_2}$), message selection
after continuation map reproduces the real fiber sampler\footnote{Uniformly from the set $\{ c :f_z(c)=m \}$ for every fixing of $z$.}; on the ideal input ($\sigma=U_{\cK}\otimes\mu_{\Lambda E_2}$), each refined ciphertext register survives with probability $1/q$. Thus $\Tr H\mu=\Tr H\sigma=1$ on the two complete comparison states.

Scaling factor need not control quantum coherences. We factor
$\Phi_m$ as $\theta\mapsto\sqrt H\theta\sqrt H$ followed by one channel.
A purification argument proves that a normalized scaled state is within
$\sqrt{2\Tr\theta(I-H)_+}$ of its unscaled state. Since
$0\preceq(I-H)_+\preceq I$, input distance $\alpha$ transfers this
deficit additively even if $\|H\|_\infty=q$. With
$d_\sigma=\Tr\sigma(I-H)_+$, the comparison becomes
\[
 \Delta(\Phi_m(\mu),\Phi_m(\sigma))
 \le\sqrt{2(d_\sigma+\alpha)}+\alpha+\sqrt{2d_\sigma}.
\]
The mean-one comparison, along with second order moments controls its error.

\heading{The same permutation counting controls noncommuting fluctuations.}
In the ideal product state, keep the full side state $\mu_{\Lambda E_2}$
and let the operators below be block diagonal in its classical label
$\Lambda$. Then, we have
\[
 \E_ZH_Z=I,\qquad
 \E_ZH_Z^2=\gamma I+(q-\gamma)\sum_cQ_c^2,
 \qquad\gamma=\frac{N-q}{N-1}.
\]
The expansion retains every ordered product $Q_cQ_{c'}$; their full sum is
$(\sum_cQ_c)^2=I$, without assuming they commute. Since
$Q_c^2\preceq\kappa Q_c$, the second moment around $I$ is at most
$q\kappa I$. Cauchy--Schwarz then gives
\[
 d_\sigma\le\tfrac12\sqrt{q\kappa},\qquad
 q\kappa=\frac{2^{d-r}}{\zeta}.
\]
Substituting into the scaling estimate gives
$\alpha+\sqrt{2\alpha}+2(2^{d-r}/\zeta)^{1/4}$.
The ideal symbol is compared with the same message-free simulator using
the classical two-point prefix distribution, at cost $\epsp$; comparing the
smoothed experiment with the original attack adds at most $\zeta$ to the
trace-distance bound. Balancing $\zeta$ with the fourth-root term gives
$\zeta\asymp2^{-(r-d)/5}$ and error
$O(\sqrt\eta+2^{-(r-d)/5})$.
No step changes the one-simulator-for-all-messages quantifier.
Figures~\ref{fig:coarse-bbj}--\ref{fig:comparison-bbj} in
Appendix~\ref{app:diagrams} display these state transitions and mark the
message-selection operation explicitly.

\subsection{Related work: rate, error, and adversary strength}
\label{sec:related}
Classical two-split NMCs progressed through multi-bit and extractor-based
constructions~\cite{ADL14,Li19}. Li~\cite{Li19} gave constant rate with
constant error; Aggarwal and Obremski~\cite{AO20} achieved constant rate with
negligible error. Aggarwal, Kanukurthi, Obbattu, Obremski, and
Sekar (AKOOS)~\cite{AKOOS22} gave near rate-$1/3$ construction. Crucially, Li's subsequent construction~\cite[Theorem
1.14]{Li23} has efficient encoding and decoding, constant rate, and error
$2^{-\Omega(\ell)}$ for an $\ell$-bit message. Our stretched-exponential
error does \emph{not} improve that classical error exponent.

The classical balanced two-split capacity $1/2$ is existential~\cite{CGcapacity}.
BBJ also gives a three-split worst-case code near rate $1/3$~\cite{BBJ23};
its extra share distinguishes it from the question resolved here.
The earlier quantum-secure, vanishing-rate worst-case code~\cite{ABJ24}
and BBJ's uniform-message two-split code do not give the combination of
properties in the last row of Table~\ref{tab:prior}. Rate-$1$ results for bounded polynomial-depth
tampering in an auxiliary-input random-oracle model~\cite{FKKS26} address
different restrictions and setup. An $\Omega(1)$ rate guarantee alone does not determine how the
achieved rate compares with $1/3$ or $1/5$.

\begin{table}[H]
\caption{Selected classical-message results. All constructions are
efficient except the capacity row. Rates are asymptotic; ``near'' means
arbitrarily close from below, and $c_{\rm AO},c_{\rm Li19},c_{\rm Li23}>0$ denote unspecified positive constants without a numerical comparison to our rate. \textbf{C}: arbitrary local classical tampering in the indicated
split-state model. \textbf{Q}: arbitrary local quantum tampering in that
model, with arbitrary pre-shared entanglement.
$\ell$ denotes message length.}
\label{tab:prior}
\centering\small\setlength{\tabcolsep}{3pt}
\begin{tabularx}{\textwidth}{@{}L{4.0cm}cL{1.8cm}>{\raggedright\arraybackslash}Xc@{}}
\toprule
Work & Splits & Rate & Messages; error & Attack\\
\midrule
Cheraghchi--Guruswami~\cite{CGcapacity} & 2 & Near $1/2$ & Worst case; negligible; nonexplicit & C\\[4pt]
Li~\cite{Li19} & 2 & $c_{\rm Li19}$ & Worst case; constant error & C\\[4pt]
Aggarwal--Obremski~\cite{AO20} & 2 & $c_{\rm AO}$ & Worst case; negligible & C\\[4pt]
AKOOS~\cite{AKOOS22} & 2 & Near $1/3$ & Worst case; $2^{-\ell^{\Omega(1)}}$; & C\\[4pt]
Li~\cite[Thm.~1.14]{Li23} & 2 & $c_{\rm Li23}$ & Worst case; $2^{-\Omega(\ell)}$ & C\\[4pt]
Batra--Boddu--Jain~\cite{BBJ23} & \textbf{3} & Near $1/3$ & Worst case; $2^{-\ell^{\Omega(1)}}$ & Q\\[4pt]
Batra--Boddu--Jain~\cite{BBJ23} & 2 & Near $1/5$ & Uniform only; $2^{-\ell^{\Omega(1)}}$ & Q\\[4pt]
\textbf{This work} & \textbf{2} & \textbf{Near $1/5$} & \textbf{Worst case; $2^{-\ell^{\Omega(1)}}$} & \textbf{Q}\\
\bottomrule
\end{tabularx}
\end{table}

\subsection{Historical development and broader context of non-malleability}\label{sec:history-codes}\label{sec:history-sharing}\label{sec:history-commitments}\label{sec:history-mpc}\label{sec:history-quantum}

Information-theoretic NMCs restrict codeword access~\cite{DDN91,DPW10,ABOsurvey},
complementing algebraic manipulation detection~\cite{AMD08} and split-state
coding~\cite{LL12,DKO13,ADL14}. Non-malleable extraction~\cite{DW09} led to non-malleable coding connections, many-tampering extensions, and improved constructions
~\cite{CGcoding14,CGL16,Li17,Li19,Li23}, while randomness encoders and
constant-rate four-state codes advanced this line~\cite{KOS18,KOS17}.
Further models cover tampering reductions, permutations, composition,
computational security, leakage, and continuous tampering
~\cite{ADKO15,AGMPP14,CL18,AAGMPR16,ADzKO15,FMNV14}; computational near-unit
rates remain compatible with information-theoretic capacity~\cite{CGcapacity}.
Our reduction instead handles prescribed-message encoding within BBJ's
quantum-secure permutation architecture~\cite{BBJ23}.

Non-malleable secret sharing extends classical privacy and reconstruction
~\cite{Shamir79,Blakley79} to tamper-resilient reconstruction
~\cite{GK18,GKgeneral18}, with work on share size, repeated tampering,
leakage, access structures, adaptive or continuous tampering, bounded joint
tampering, local reconstruction, and short shares
~\cite{BS19,SV19,ADNOPRS19,FV19,BFV19,BFOSV20,LCGSW19,KOST21,CKOS22}.
Non-malleable commitments control related committed values~\cite{DDN91,PR05}
under several notions~\cite{PR05concurrent,LPV08}, with constructions from
setup and constant-round protocols through simulation-extractability,
amplification, and one-way-function assumptions
~\cite{DIO98,FF00,DG03,Barak02,PR05,PR05concurrent,LPV08,LP09,LP11}. Non-malleable codes also support commitments~\cite{AGMPP14,GPR16} and later round-efficient
concurrent protocols~\cite{COSV16,COSV17,GR19}.

Multi party computation (MPC) is developed from foundational computational and information-theoretic
protocols~\cite{Yao82,GMW87,BGW88,CCD88} to constant-round
constructions~\cite{BMR90}. Non-malleability and amplification support input
control and round efficiency~\cite{PR05,LP11,Wee10,GPR16}, while
leakage-resilient sharing supports leakage-tolerant MPC~\cite{SV19}.
Composability requires stronger simulation tools
~\cite{Canetti01,CF01,CLOS02,LP12,Kiy14}; plain-model four-round MPC and
oblivious transfer based work sharpen tradeoffs among rounds, setup, black-box use, and
assumptions~\cite{BHP17,GS18,BL18,CCGJO20,IKSS23,COSW25,HMS26}. These applications require guarantees beyond those proved here.

Post-quantum interactive non-malleability and entanglement-assisted local
tampering are distinct. Post-quantum commitments and secure computation
achieve low round complexity under computational assumptions
~\cite{BLS22,LPY23,CLPY25}, with some MPC applications additionally using
quantum fully homomorphic encryption and quantum-hard learning with errors assumption. For non-malleable coding, Aggarwal, Boddu and Jain~\cite{ABJ24} and BBJ~\cite{BBJ23} protect classical data against entangled local tampering; other work treats quantum messages, secret sharing, or tamper detection under additional entanglement restrictions~\cite{BGJR24,BB25}. Our result keeps messages and codewords classical, allows arbitrary local quantum computation with pre-shared entanglement in the split-state model.

\heading{Organization.}
Section~\ref{sec:prelim} introduces the preliminaries and the precise BBJ
guarantee used in the proof. Section~\ref{sec:construction} gives the construction
and finite-parameter security theorem; Sections~\ref{sec:rearrangement}--\ref{sec:security} prove it.
Section~\ref{sec:parameters} gives rates and a balanced-share corollary.
Appendix~\ref{app:classical} is the complete classical proof;
Appendix~\ref{app:quantum-facts} contains standard quantum facts, and
Appendix~\ref{app:diagrams} gives the register diagrams.

\section{Preliminaries}
\label{sec:prelim}

\subsection{Classical distributions and quantum states in registers}
All logarithms are binary. Quantum registers are finite-dimensional;
arbitrary pre-shared entanglement means that any finite local dimensions
are allowed, with a security error independent of those dimensions.
We write $(\LeftShare,\RightShare)$ for codeword shares and reserve $R$ for random
padding. Classical variables use uppercase letters and
their values lowercase: tampering maps $X,(Y,C)$ to $X',(Y',C')$.
Following BBJ, $E_1,E_2$ name the auxiliary quantum registers on the two
sides, and $E'_1,E'_2$ their outputs. The square-root coarse operation acts
on $E_2$ itself; continuation produces $E'_2$. The state in a register changes
even when its register name is reused. The local quantum channels $\cA_1,\cA_2$ are the
channel forms of BBJ's local isometries $U,V$. We retain $Z$ for the extractor
output and $C$ for the ciphertext (BBJ uses $R,Z$, respectively), so that $R$
continues to denote our random padding. Appendix~\ref{app:diagrams} gives the
notation correspondence and the register diagrams.
$U_k$ is uniform on $k$ bits; $U_A$ is uniform on a finite set $A$.

A classical channel is a conditional distribution
$P_{B\mid A}(b\mid a)=\Pr[B=b\mid A=a]$. With input distribution $P_A$, it gives
$P_{AB}(a,b)=P_A(a)P_{B\mid A}(b\mid a)$; changing only the input distribution
leaves the channel unchanged. Superscript ${\rm hyb}$ denotes the
uniform-ciphertext hybrid. Conditional distributions at zero-probability input
values may be completed arbitrarily.

A quantum register $E_2$ is described by its state, not by a classical
conditional probability. The full classical--quantum state is
\[
 \rho_{AE_2}=\sum_a P_A(a)\ket a\bra a\otimes\rho_{E_2\mid A=a}
          =\sum_a\ket a\bra a\otimes\omega_{E_2}^a,
 \qquad \omega_{E_2}^a=P_A(a)\rho_{E_2\mid A=a}.
\]
Each block $\omega_{E_2}^a$ has trace $P_A(a)$, and the sum of these
traces is one. A zero-probability block is zero, regardless of its chosen conditional
state. We work with these joint states instead of repeatedly normalizing
individual events. Write $\Delta(\rho,\sigma)=\tfrac12\|\rho-\sigma\|_1$,
which is total variation on diagonal classical states. We identify a
classical distribution $P_A$ with the diagonal density matrix
$\sum_a P_A(a)\ket a\bra a$ whenever it is used as a quantum state.
The notation $P_J$ denotes the law of a classical random variable $J$;
a randomized experiment appearing in a distributional comparison denotes
its output law. 

A completely positive (CP) map has Kraus form
$\Phi(A)=\sum_j L_j A L_j^\dagger$. Its adjoint satisfies
$\Tr[B\Phi(A)]=\Tr[\Phi^*(B)A]$, so
$\Phi^*(I)=\sum_j L_j^\dagger L_j$ determines the output trace.
A quantum channel is completely positive and trace-preserving (CPTP);
for a CP map $\Phi$, this is equivalent to $\Phi^*(I)=I$.
An instrument is a family of CP maps whose sum is a channel; each of its CP maps registers an outcome's probability and its update of the state in the auxiliary output register.
Complete positivity alone does not imply trace preservation or nonincrease. In particular,
our rescaled message-selection maps need not be physical outcome CP maps.
Appendix~\ref{app:quantum-facts} proves the following tools and the
Kraus representation.
\begin{lemma}[Trace-distance tools]\label[lemma]{lem:tools}
For normalized states,
\[
 \Delta(\rho,\sigma)=\max_{0\preceq B\preceq I}|\Tr B(\rho-\sigma)|.
\]
CP trace-nonincreasing maps reduce the trace norm of Hermitian inputs.
The trace norm of a classical block-diagonal matrix is the sum of its
blocks' trace norms. Tensoring with a fixed normalized state preserves
trace distance. For a Hermitian $A$ and normalized state $\theta$,
\[
 \Tr\theta|A|\le\sqrt{\Tr\theta A^2}.
\]
For pure unit vectors $u,v$, their density matrices have distance
$\sqrt{1-|\langle u,v\rangle|^2}$; the distance between any reductions is
no greater.
\end{lemma}
Here $A_+$ is the spectral positive part and $|A|=\sqrt{A^2}$. Neither is
an entrywise operation. If $B\succeq A$ and $\theta\succeq0$, then
$\Tr\theta B\ge\Tr\theta A$, since
$\Tr\theta(B-A)=\Tr[\sqrt\theta(B-A)\sqrt\theta]\ge0$.
We use this rule rather than assuming a product of positive matrices is
positive. Conjugation also preserves positive order.

\Needspace{22\baselineskip}
\heading{Notation at a glance.}
The following symbols are used throughout; their constructions and
state definitions are given in the indicated sections.
\begin{center}
\small
\begin{tabularx}{\linewidth}{@{}L{4.6cm}>{\raggedright\arraybackslash}X@{}}
\toprule
Symbol & Meaning\\\midrule
$\LeftShare,\RightShare$; $R$ & Left and right shares; fresh random padding\\[3pt]
$n,d$; $\ell,r,s=\ell+r$ & Source lengths; message, padding, and ciphertext lengths\\[3pt]
$q=2^\ell,T=2^r,N=2^s=qT$ & Message count, fiber size, and field size\\[3pt]
$s_0,2s$ & Full and retained $\nmext$ output lengths (Section~\ref{sec:bbj-interface})\\[3pt]
$g=r-d$ & Padding surplus (Section~\ref{sec:parameters})\\[3pt]
$\cK=\F_N^\times\times\F_N$ & Ideal key set (Section~\ref{sec:construction})\\[3pt]
$S,\overline Z'$; $\sflag$ & Source-equality flag and masked tampered key; a fixed flag value\\[3pt]
$\Lambda=(Y,Y',S,\overline Z')$ & retained side information\\[3pt]
$\mu,\sigma$ & Actual and ideal-key coarse states (Section~\ref{sec:rearrangement})\\[3pt]
$\nu$ & State of the continued side register on $\Lambda CC'E'_2$ (Section~\ref{sec:rearrangement})\\[3pt]
$\kappa=2^{d-s}/\zeta$ & Continuation map bound in \eqref{eq:Q}\\[3pt]
$\Phi_m$; $H=\Phi_m^*(I)$ & Message-selected output map and its output-scale operator (Section~\ref{sec:security})\\[3pt]
$\eta,\zeta,\alpha,\beta,\epsp$ & BBJ security error, smoothing parameter, input- and output-comparison bounds, and finite-permutation correction (Section~\ref{sec:construction})\\
\bottomrule
\end{tabularx}
\end{center}

\Needspace{10\baselineskip}
\subsection{Two-split quantum tampering and non-malleability}
\begin{definition}[Local quantum attack]\label[definition]{def:attack}
An attack on classical shares $(\LeftShare,\RightShare)$ consists of a bipartite state
$\psi_{E_1E_2}$, independent of the message and the encoder's randomness,
and two local channels
\[
 \cA_1:\LeftShare E_1\longrightarrow \LeftShare'E'_1,\qquad
 \cA_2:\RightShare E_2\longrightarrow \RightShare'E'_2.
\]
The channels have no communication link. Output shares are classical
strings in the prescribed alphabets. The local auxiliary-register dimensions are
arbitrary, and the security error must be independent of
their dimensions.
\end{definition}
A classical decoder can begin by measuring any purported quantum output
share in its fixed computational basis, so requiring classical returned
shares does not restrict the effective attacks against this decoder.
Classical source copies can be kept in the analysis without disturbing the
input.

\begin{definition}[Worst-case two-split NMC]\label[definition]{def:nmc}
A randomized encoder and a deterministic decoder with types
\[
 \Enc:\bits^\ell\longrightarrow\bits^{n_L}\times\bits^{n_R},\qquad
 \Dec:\bits^{n_L}\times\bits^{n_R}\longrightarrow\bits^\ell\cup\{\perp\}
\]
have perfect correctness when
$\Dec(\Enc(m))=m$ with probability one for every $m$. They are an
$\varepsilon$-secure NMC if, for every fixed attack $\cA$ from
\cref{def:attack}, there exists a distribution $D_{\cA}$ on
$\bits^\ell\cup\{\same,\perp\}$ such that
\begin{equation}
 \forall m\in\bits^\ell:\quad
 \Delta\bigl(\Dec(\cA(\Enc(m))),\Copy(m,D_{\cA})\bigr)\le\varepsilon.
 \label{eq:nmc-def}
\end{equation}
Here $\Copy(m,j)=m$ if $j=\same$ and $\Copy(m,j)=j$ otherwise.
For a distribution $D$, $\Copy(m,D)$ denotes the law of $\Copy(m,J)$ when
$J\sim D$. The rate is $\ell/(n_L+n_R)$, where $n_L,n_R$ are the lengths
of $\LeftShare,\RightShare$, respectively.
\end{definition}
Equation~\eqref{eq:nmc-def} compares decoded-message distributions; it
does not assert a joint simulation of the adversary's retained quantum
registers. The quantifier order $\forall\cA\,\exists D_{\cA}\,\forall m$ requires
one simulator distribution for all messages. The simulator may use the
attack's specified shared state; it is not an efficient classical simulation
of an unbounded quantum computation. Our decoder never outputs
$\perp$ on a prescribed-length word. 

\subsection{The precise BBJ security we invoke}
\label{sec:bbj-interface}
Let
$\nmext_0:\bits^n\times\bits^d\to\bits^{s_0}$ be deterministic.
We will retain $2s\le s_0$ fixed output coordinates and call the resulting
function $\nmext$. In the specification below write $k_Z$ for a general
output length. On independent uniform sources let
\[
 \begin{gathered}
 Z=\nmext(X,Y),\qquad Z'=\nmext(X',Y'),\qquad
 S=\ind[(X',Y')=(X,Y)],\\
 \overline Z'=\begin{cases} \same,&S=1,\\Z',&S=0.\end{cases}
 \end{gathered}
\]

\begin{definition}[Augmented non-malleability guarantee~\cite{BBJ23}]\label[definition]{def:inner}
For every local quantum attack on $X,Y$, with source-independent shared
entanglement, let $E'_2$ denote its retained right output register. The security guarantee with error $\eta$ is
\begin{equation}
 \Delta\bigl(\rho_{ZS\overline Z'YY'E'_2},
       U_{k_Z}\otimes\rho_{S\overline Z'YY'E'_2}\bigr)\le\eta.
 \label{eq:inner}
\end{equation}
\end{definition}
The flag $S$ is computed in the analysis using reserved original values. It is
not a message exchanged by the attackers. On $S=1$, we use the symbol $\same$ to denote $Z'=Z$. A
collision $Z'=Z$ on changed source inputs remains in the case $S=0$. 

\begin{assumption}[Batra--Boddu--Jain~\cite{BBJ23}]\label[assumption]{ass:bbj}
For every sufficiently small fixed $\delta>0$ allowed by the BBJ construction,
the construction is polynomial-time computable and satisfies
\cref{def:inner} with
\begin{equation}
 d=\delta n+O(1),\quad s_0=(1/2-\delta)n+O(1),\quad
 \eta=2^{-n^{\Omega(1)}}.
 \label{eq:bbj-params}
\end{equation}
The sources are independent and uniform; adversaries may use arbitrary
shared entanglement and local isometries. 
\end{assumption}

\heading{From BBJ's scaled guarantee to \cref{def:inner}.}
We use the right-residual scaled inequality of BBJ~\cite[Lemma 3]{BBJ23}.
All references to BBJ's lemma, theorem, and figure numbers refer to the
2023 preprint version~1 identified in the bibliography.
A mixed shared state can be purified by assigning its purifying register
to the left attacker. Each local channel can be dilated to an isometry:
for Kraus operators $L_j$, the map $v\mapsto\sum_j L_jv\otimes\ket j$
is isometric because $\sum_jL_j^\dagger L_j=I$. Tracing out the added
register recovers the channel. The right isometry may also keep copies
of its classical input $Y$ and classical output $Y'$ in its residual
register. Tracing out unused dilation registers therefore gives the
retained register $B=YY'E'_2$ required here. These are local operations;
they do not disclose $X$ or $X'$ to the right attacker.

Write $Z_0=\nmext_0(X,Y)$ and $Z'_0=\nmext_0(X',Y')$.
Let $p=\Pr[S=1]$, and let $\rho^1,\rho^0$ denote the normalized
same-input and changed-input branch states, respectively. BBJ's scaled
inequality, followed by the indicated local discarding, gives
\[
\begin{aligned}
 &p\normone{\rho^1_{Z_0B}-U_{s_0}\otimes\rho^1_B}\\
 &\quad +(1-p)\normone{\rho^0_{Z_0Z'_0B}
                    -U_{s_0}\otimes\rho^0_{Z'_0B}}
 \le 4(2^{-k}+\varepsilon_{\rm BBJ}),
\end{aligned}
\]
where $k$ and $\varepsilon_{\rm BBJ}$ are BBJ's parameters. A branch
of probability zero contributes zero and needs no conditional state.
Attach the classical flag $S$ and replace $Z'_0$ by $\same$ on $S=1$.
The trace norm of the resulting difference is exactly the sum of the
two disjoint branch norms above. Since our trace distance is half the
trace norm, this proves \cref{def:inner} with
$\eta=2(2^{-k}+\varepsilon_{\rm BBJ})$. For BBJ's allowed polynomially
growing $k$ and stretched-exponentially small $\varepsilon_{\rm BBJ}$,
this has the form in \eqref{eq:bbj-params}. Finally, applying the same
fixed coordinate projection to $Z_0$ and to the unmasked $Z'_0$ is
deterministic processing. It preserves the error and yields a uniform
$2s$-bit original key, so the truncated function also satisfies
\cref{def:inner}, with $k_Z=2s$. No left-residual or simultaneous
both-residual guarantee is used.

\section{Construction and exact permutation facts}
\label{sec:construction}
Fix positive integers $n,d,\ell,r$ and put
\begin{equation}
 s=\ell+r,\quad N=2^s,\quad q=2^\ell,\quad T=2^r,
 \qquad N=qT,
 \label{eq:length-names}
\end{equation}
with $2s\le s_0$. Work over $\F=\F_N$ in a fixed binary representation.
A concatenation $m\|R$ is interpreted as its field element, and the prefix
operation uses the same representation. 

\subsection{A total bit-seeded affine family}
For a $2s$-bit seed $z=(a,b)\in\F^2$, define
\begin{equation}
 a^\star=\nz(a)=\begin{cases}a,&a\ne0,\\1,&a=0,\end{cases}
 \quad \pi_z(u)=a^\star u+b,
 \quad f_z(c)=\pref\bigl((a^\star)^{-1}(c-b)\bigr).
 \label{eq:permutation}
\end{equation}
The ideal seed set is $\cK=\F^\times\times\F$. Uniform seeds and
$U_{\cK}$ are different distributions on the same seed alphabet.
Write
\[
 u_{\cK}(z):=\Pr_{Z\sim U_{\cK}}[Z=z]
           =\frac{\ind[z\in\cK]}{N(N-1)},\qquad z\in\F^2,
\]
for the ideal key's probability mass function on the raw seed alphabet.

\begin{lemma}[Seed distribution and two-point uniformity]\label[lemma]{lem:affine}
Every fixed seed defines a permutation. Moreover,
$\Delta(U_{2s},U_{\cK})=1/N$. Under $Z\leftarrow U_{\cK}$,
for fixed distinct ciphertexts $c,c'$, the ordered inverse pair
$(\pi_Z^{-1}(c),\pi_Z^{-1}(c'))$ is uniform over the $N(N-1)$ ordered
distinct field pairs.
\end{lemma}
\begin{proof}
Since $a^\star\ne0$, $\pi_z$ is a permutation with inverse
$\pi_z^{-1}(c)=(a^\star)^{-1}(c-b)$ for every $a,b$. The decoder $f_z$
keeps only the first $\ell$ bits of this inverse image. Uniform seeds have
$a=0$ with probability
$1/N$; conditioning on its complement produces $U_{\cK}$, with distance
exactly $1/N$. Equivalently, sum the absolute probability differences on
and off that event.
For any $u\ne u'$ and $c\ne c'$, the equations
$au+b=c$, $au'+b=c'$ have the unique solution
$a=(c-c')/(u-u')\ne0$ and $b=c-au$. Thus exactly one ideal key gives each
ordered inverse pair, proving uniformity.
\end{proof}

\begin{lemma}[Fibers and conditional prefixes]\label[lemma]{lem:fibers}
For every fixed seed and every message, the set
$\mathcal F_z(m)=\{c:f_z(c)=m\}$ has exactly $T$ elements, sampled
uniformly by $\pi_z(m\|U_r)$. For an ideal key, fixed message $m$, and
$c\ne c'$,
\begin{align}
 \Pr_Z[f_Z(c)=m]&=1/q,\label{eq:onepoint}\\
 \Pr_Z[f_Z(c)=m,f_Z(c')=m]&=\frac{T(T-1)}{N(N-1)}.
 \label{eq:twomembership}
\end{align}
The conditional distribution of $f_Z(c')$ given $f_Z(c)=m$ is
\begin{equation}
 p_m(u)=\begin{cases}(T-1)/(N-1),&u=m,\\T/(N-1),&u\ne m,\end{cases}
 \quad
 \Delta(p_m,U_\ell)=\epsp=\frac{q-1}{q(N-1)}.
 \label{eq:prefix-law}
\end{equation}
\end{lemma}
\begin{proof}
The $T$ strings with prefix $m$ map bijectively to $\mathcal F_z(m)$.
Their uniform images therefore give the stated sampler for every fixed key.
For an ideal key, an inverse image is uniform on $N$ elements, of which
$T$ have prefix $m$. This proves \eqref{eq:onepoint}. There are $T(T-1)$
ordered distinct inverse pairs with both prefixes $m$, proving
\eqref{eq:twomembership}. For a different second prefix $u$ there are
$T^2$ such pairs. Divide the corresponding joint probabilities by $T/N$
to obtain the conditional distribution.
Its probabilities sum to one because $(T-1)+(q-1)T=N-1$.
At $u=m$ its deficit from $1/q$ is $(q-1)/(q(N-1))$; at each other
message its excess is $1/(q(N-1))$. Summing positive differences proves
the exact total variation formula.
\end{proof}
Joint fiber membership will control the message-selection second moment;
the conditional second prefix will control same-key ciphertext changes.

\subsection{Encoding and decoding}
\begin{construction}[Padding and a single affine permutation]\label[construction]{con:code}
To encode $m\in\bits^\ell$, sample $X\leftarrow U_n$, $Y\leftarrow U_d$,
and $R\leftarrow U_r$ independently. Compute the truncated output
$Z=(a,b)=\nmext(X,Y)$, form $C=\pi_Z(m\|R)$, and output
\begin{equation}
 \LeftShare=X,\qquad \RightShare=(Y,C).
 \label{eq:codeword}
\end{equation}
On $X',(Y',C')$, compute $Z'=\nmext(X',Y')$ and output $f_{Z'}(C')$.
Here $R$ is the random padding, while $\RightShare$ denotes the right share.
\end{construction}

\begin{lemma}[Correctness and complexity]\label[lemma]{lem:correctness}
The code has perfect correctness, uses exactly $n+d+r$ independent random
bits in encoding, and has rate $\ell/(n+d+s)$. Given the public field representation, encoding and decoding run in time polynomial in $n+d+s$.
\end{lemma}
\begin{proof}
The decoder recomputes the same key and
$(a^\star)^{-1}(\pi_Z(m\|R)-b)=m\|R$, also when $a=0$. The first
$\ell$ bits are $m$ for every execution. $\nmext$ evaluation is polynomial time 
by \cref{ass:bbj}; field addition, multiplication, and inversion are
polynomial-time operations.\footnote{Using a publicly fixed irreducible degree-$s$ polynomial over $\F_2$
as modulus, field addition, multiplication, and inversion of nonzero
elements take polynomial time in $s$. The representation uses no
per-message randomness.} The only random bits are the three sampled
strings. The transmitted lengths are $n$ and $d+s$, proving the rate.
\end{proof}

\subsection{The exact uniform-ciphertext hybrid}
For a fixed attack, define a hybrid by sampling $X,Y,C$ independently
uniform, computing $Z$ for the analysis, and applying the two local
operations. Its latent message is $M=f_Z(C)$.
\begin{lemma}[Exact real--hybrid identity]\label[lemma]{lem:real-hybrid}
In this hybrid, $M$ is uniform and independent of $X,Y$. For every $m$, the
hybrid conditioned on $M=m$ is exactly the real encoding-and-tampering
experiment, including retained quantum output.
\end{lemma}
\begin{proof}
For each $x,y$, exactly $T$ of the $N$ ciphertexts have the requested
prefix under its total permutation. Therefore $\Pr[M=m\mid x,y]=1/q$.
Write $z(x,y)=\nmext(x,y)$ and $p(x,y)=2^{-n-d}$. Bayes' rule gives
\begin{equation}
 \Pr[x,y,c\mid M=m]
 =\frac{p(x,y)}{T}\ind[f_{z(x,y)}(c)=m].
 \label{eq:real-law}
\end{equation}
This is the source distribution times the exact fiber sampler in
\cref{lem:fibers}. For each fixed triple $(x,y,c)$, the same attack on
the same shared state has the same conditional quantum output in both
experiments. Multiplying those matrices by \eqref{eq:real-law} and summing
proves equality of the entire states.
\end{proof}
For any output register $J$ jointly distributed with a uniform classical
$M\in\bits^\ell$, block-diagonal additivity gives
\[
 \Delta(\rho_{MJ},U_\ell\otimes\rho_J)
 =q^{-1}\sum_{m\in\bits^\ell}
       \Delta(\rho_{J\mid M=m},\rho_J).
\]
The following reduction controls each fixed-message output without the
generic factor $q$ associated with conditioning an average error bound.

\subsection{The main theorem}
We now state the main theorem. Its proof occupies Sections~\ref{sec:rearrangement}--\ref{sec:security}.

\begin{theorem}[Quantum-secure two-split code]\label[theorem]{thm:main}
Let $n,d,\ell,r$ be positive integers, $s=\ell+r$, and $2s\le s_0$.
Assume the truncated $\nmext$ satisfies \cref{def:inner} with error
$\eta$. For $0<\zeta<1$, put
\begin{align}
 \alpha&=\eta+2^{-s},\label{eq:alpha}\\
 \beta&=\alpha+\sqrt{2\alpha}
               +2\left(\frac{2^{d-r}}{\zeta}\right)^{1/4},
               \label{eq:beta}\\
 \epsp&=\frac{2^\ell-1}{2^\ell(2^s-1)}.
 \label{eq:pair-error}
\end{align}
Then \cref{con:code} is a perfectly correct, worst-case two-split quantum-secure
NMC of rate $\ell/(n+d+s)$ and error
\begin{equation}
 \epsn\le\zeta+\beta+\epsp.
 \label{eq:main-error}
\end{equation}
If $g=r-d>0$, choosing
$\zeta=2^{-\lceil g/5\rceil}$ gives, whenever $\alpha\le1$,
\begin{equation}
 \epsn=O\bigl(\sqrt\eta+2^{-g/5}\bigr).
 \label{eq:main-asymptotic}
\end{equation}
\end{theorem}

\section{Reordering the attack before message conditioning}
\label{sec:rearrangement}

Fix an arbitrary quantum attack. We reorder its uniform-ciphertext
experiment so that BBJ's non-malleable randomness-encoder guarantee can be
applied to an intermediate state while retaining the right auxiliary
register $E_2$. The factorization is an equality of quantum operations.

\subsection{Smoothing and invertible probability operators}
Fix $0<\zeta<1$. Modify only the right operation: with probability
$1-\zeta$ run the original operation; with probability $\zeta$ choose
$Y'$ uniformly, keep $C'=C$, and retain the right input state in $E_2$. Call the
smoothed attack $\cA_\zeta$. Embed the two possible auxiliary output
registers in a common direct-sum space with a classical flag registering
the branch. In this quantum proof,
$P^{\rm hyb}$ refers to the uniform-ciphertext hybrid for $\cA_\zeta$.

Let $\mathcal V^{y,c}_{y',c'}:E_2\to E'_2$ be the CP outcome maps of
the smoothed right attack.
\begin{lemma}[Smoothing cost]\label[lemma]{lem:fallback}
For every input, the original and smoothed output states are within
$\zeta$ in trace distance, also after decoding. For all classical values
$y,c$ and $y'\in\bits^d$, the probability operator for the event $Y'=y'$,
after summing over $c'$, satisfies
\[
 M_{y'}^{y,c}:=\sum_{c'}(\mathcal V^{y,c}_{y',c'})^*(I)
             \succeq\zeta2^{-d}I.
\]
\end{lemma}
\begin{proof}
The new state is $(1-\zeta)\rho+\zeta\theta$, compared with the original
embedded state $\rho$. Their distance is
$\zeta\Delta(\theta,\rho)\le\zeta$, and decoding cannot increase it.
The smoothing selects each $y'$ independently of the state in $E_2$ with probability
$2^{-d}$; its scaled probability operator is $\zeta2^{-d}I$.
After summing over $c'$, the original branch contributes a positive
operator. This proves the lower bound for the aggregated event $Y'=y'$.
\end{proof}
Smoothing is used only in the proof and simulator, to make these
aggregated probability operators and their averages invertible.

For a joint input $\rho_{YCE_2}=\sum_{y,c}\ket{y,c}\bra{y,c}\otimes\omega_{E_2}^{y,c}$,
the output with classical input copies retained is
\[
 \sum_{y,c,y',c'}\ket{y,c,y',c'}\bra{y,c,y',c'}\otimes
              \mathcal V^{y,c}_{y',c'}(\omega_{E_2}^{y,c}).
\]
This describes all outcomes together, including their output states in register $E'_2$.
Their probabilities are the traces of the displayed blocks; no division by an
outcome probability is needed. Define
\begin{equation}
 M_{y'}^{y,c}=\sum_{c'}(\mathcal V^{y,c}_{y',c'})^*(I),
 \qquad G_{y'}^y=\frac1N\sum_cM_{y'}^{y,c}.
 \label{eq:MG}
\end{equation}
For every $y,c$, the fine CP maps sum to a channel, so
\begin{equation}
 \sum_{y'} M_{y'}^{y,c}=I,\quad
 \zeta2^{-d}I\preceq M_{y'}^{y,c}\preceq I,
 \quad \sum_{y'}G_{y'}^y=I,\quad G_{y'}^y\succeq\zeta2^{-d}I.
 \label{eq:Gbounds}
\end{equation}
The scalar $\Tr G_{y'}^y\theta$ is the hybrid probability of the output
$y'$ on input state $\theta_{E_2}$. There is no dependence on the original key:
the actual right procedure receives only $y,c$ and the state in $E_2$.

\subsection{An exact continuation after the output \texorpdfstring{$Y'$}{Y'}}
\begin{lemma}[Continuation factorization]\label[lemma]{lem:continuation}
For each fixed $y$, the CP maps
\begin{equation}
 \mathcal C_{y'}^y(\theta)=\sqrt{G_{y'}^y}\theta\sqrt{G_{y'}^y}
 \label{eq:coarse}
\end{equation}
indexed by $y'$ form an instrument. For each fixed pair $(y,y')$, define
\begin{equation}
 \mathcal R^{y,y'}_{c,c'}(\theta)=\frac1N
 \mathcal V^{y,c}_{y',c'}\bigl((G_{y'}^y)^{-1/2}\theta(G_{y'}^y)^{-1/2}\bigr).
 \label{eq:continuation}
\end{equation}
For this fixed pair, the sum over $c,c'$ is a channel and
\begin{equation}
 \mathcal R^{y,y'}_{c,c'}\circ\mathcal C_{y'}^y
       =\frac1N\mathcal V^{y,c}_{y',c'}.
 \label{eq:composition}
\end{equation}
Its probability operators for $C=c$ obey
\begin{equation}
 Q_c^{y,y'}:=\sum_{c'}(\mathcal R^{y,y'}_{c,c'})^*(I),\qquad
 \sum_cQ_c^{y,y'}=I,\qquad
 0\preceq Q_c^{y,y'}\preceq\kappa I,
 \quad \kappa=\frac{2^{d-s}}{\zeta}.
 \label{eq:Q}
\end{equation}
\end{lemma}
\begin{proof}
Conjugation, positive scalar multiplication, and composition preserve
complete positivity. For each fixed $y$, the CP maps have probability
operators $G_{y'}^y$, which sum to $I$ over $y'$ and hence form a positive
operator-valued measure (POVM). For the continuation indexed by $(y,y')$, abbreviate
$G=G_{y'}^y$. Its total conjugate map is
\[
 \sum_{c,c'}(\mathcal R^{y,y'}_{c,c'})^*(I)
 =G^{-1/2}\left(\frac1N\sum_cM_{y'}^{y,c}\right)G^{-1/2}
 =G^{-1/2}GG^{-1/2}=I.
\]
Thus the continuation is trace-preserving on \emph{every} input state,
not merely on states produced by the preceding coarse map. Substituting
$\sqrt G\theta\sqrt G$ cancels adjacent inverse square roots and proves
\eqref{eq:composition}.

The individual POVM operator is
\[
 Q_c^{y,y'}=\frac{1}{N}G^{-1/2}M_{y'}^{y,c}G^{-1/2}
 \preceq \frac{1}{N}G^{-1}.
\]
Since every eigenvalue of $G$ is at least $\zeta2^{-d}$,
$G^{-1}\preceq2^d\zeta^{-1}I$, proving \eqref{eq:Q}. The sum assertion
follows from the trace-preserving calculation.
\end{proof}

The CP map is a chosen square-root realization of the POVM corresponding to the event
$Y'=y'$; its output state in $E_2$ need not equal the original attack's
state in $E'_2$ conditioned on that event. The continuation factorization corrects that difference.
The rearrangement reproduces the \emph{uniform mixture} over original
ciphertexts, not the original attack on an arbitrary fixed preexisting
$c$. In a one-dimensional auxiliary register, the same identity is Bayes' rule:
\[
 P^{\rm hyb}_{C\mid Y,Y'}(c\mid y,y')=
 \frac{P^{\rm hyb}_{Y'\mid Y,C}(y'\mid y,c)}
      {N P^{\rm hyb}_{Y'\mid Y}(y'\mid y)}.
\]
Here the smoothing ensures a positive denominator. The quantum formula
preserves the output state in $E'_2$ as well as this scalar probability.

Equation~\eqref{eq:composition} holds on all matrices, so tensoring it with
an identity map on any reference preserves it. Equivalently, expand a
bipartite matrix as $\sum_{i,j}\ket i\bra j\otimes T_{ij}$ and apply
the identity to each $T_{ij}$. This proves reconstruction of entangled
correlations, which equality of outcome probabilities alone would not prove.
The alphabet of $Y'$ contributes the factor $2^d$ in \eqref{eq:Q},
while averaging over the ciphertext contributes the factor $2^{-s}$.

\subsection{The actual intermediate state and its ideal key}
Let $\{\mathcal A^x_{x'}\}_{x'}$ describe the left outcome CP maps.
After applying them to the independent uniform source $X$ and shared state
$\psi$, and discarding the left private output, the joint state is
\[
 \rho_{XX'E_2}=2^{-n}\sum_{x,x'}\ket{x,x'}\bra{x,x'}\otimes\omega_{E_2}^{x,x'},
 \quad
 \omega_{E_2}^{x,x'}=\Tr_{E'_1}
                  [(\mathcal A^x_{x'}\otimes\Id_{E_2})(\psi)].
\]
The blocks include the left-outcome probabilities and
$\sum_{x'}\omega_{E_2}^{x,x'}=\psi_{E_2}$ for every $x$. Sample $Y\leftarrow U_d$ independently, append it, and
apply the complete coarse procedure. Compute $Z,S,\overline Z'$ as
classical functions. Writing $z=z(x,y)=\nmext(x,y)$,
$\sflag(x,y,x',y')=\ind[(x',y')=(x,y)]$, and
\[
 \lambda(x,y,x',y')=\bigl(y,y',\sflag(x,y,x',y'),\overline z'(x,y,x',y')\bigr)
\]
gives the joint state
\begin{equation}
 \mu_{Z\Lambda XX' E_2}
 =\sum_{x,y,x',y'}2^{-n-d}\ket{z,\lambda, x, x'}\bra{z,\lambda, x, x'}
       \otimes\mathcal C_{y'}^y(\omega_{E_2}^{x,x'}).
 \label{eq:mu}
\end{equation}
It is normalized and does not depend on a requested message. The sum over
coarse CP maps preserves the original source distribution. For the BBJ comparison we use the marginal
\[
 \mu:=\mu_{Z\Lambda E_2}
     =\Tr_{XX'}\mu_{Z\Lambda XX'E_2}.
\]
Thus the displayed copies of $X,X'$ are discarded; they are \emph{not}
extra registers disclosed in the BBJ security guarantee.

For every summand, continuing the state in its right register produces
\[
 2^{-n-d}\mathcal R^{y,y'}_{c,c'}\bigl(\mathcal C_{y'}^y(\omega_{E_2}^{x,x'})\bigr)
 =2^{-n-d}N^{-1}\mathcal V^{y,c}_{y',c'}(\omega_{E_2}^{x,x'}),
\]
exactly the corresponding hybrid block. All summands with the same
$z,\lambda$ have the same continuation control $y,y'$, so this remains
true after grouping and discarding $X,X'$. The flag and masked output can
be computed before continuation because they depend only on
$x,y,x',y'$. A message test depends on the later $c$ and is not included in
these coarse CP maps.

The local instruments $\{\mathcal A^x_{x'}\}_{x'}$ and $\{\mathcal C^y_{y'}\}_{y'}$ define a valid attack, so we can invoke the BBJ security of $\nmext$. Apply \eqref{eq:inner} once and use \cref{lem:affine} to obtain
\begin{equation}
 \sigma_{Z\Lambda E_2}=U_{\cK}\otimes\mu_{\Lambda E_2},\qquad
 \Delta(\mu,\sigma)\le\alpha:=\eta+N^{-1}.
 \label{eq:ideal}
\end{equation}
This product comparison state need not be a
consistent encoding; its independence permits ideal-key averaging. The message test has not yet been applied: it needs only the retained $Z$ and the $C$ produced by continuation factorization.

Let $\mathcal R$ denote continuation with its classical outcomes retained,
controlled by $y,y'$ inside $\lambda$. Define
\begin{equation}
 \nu=\mathcal R(\mu_{\Lambda E_2}).
 \label{eq:nu}
\end{equation}
It is exactly the smoothed uniform-ciphertext register $(\Lambda,C,C',E'_2)$.
In particular, it is normalized and independent of the requested message.
It can be sampled by a quantum procedure that runs the actual smoothed
hybrid, without implementing inverse square roots.

\section{Mean-one scaling and the permutation second moment}
\label{sec:weighting}
The message-selection operation will be a CP map but not trace-preserving on
arbitrary inputs: it selects a classical event and multiplies by $q$.
Ordinary channel data processing cannot be used. This section proves the comparison that takes its place, then supplies its ideal-state deficit.

\subsection{A positive scaling theorem}
\begin{lemma}[Factorization through the output-scale operator]\label[lemma]{lem:effect-factor}
For any finite-dimensional CP map $\Phi$ with $H=\Phi^*(I)$, there is one
channel $\Psi$ such that, on every input matrix $A$,
\begin{equation}
 \Phi(A)=\Psi(\sqrt H A\sqrt H).
 \label{eq:effect-factor}
\end{equation}
\end{lemma}
\begin{proof}
Let $L_j$ be Kraus operators, so $H=\sum_j L_j^\dagger L_j$. Let $P$
project onto the support of $H$. If $e\in\ker H$, then
$0=\langle e,He\rangle=\sum_j\|L_je\|^2$, so $L_j=L_jP$.
Define $H^{-1/2}$ to be the inverse on its support and zero on its kernel,
and put $R_j=L_jH^{-1/2}$. Then
$\sum_jR_j^\dagger R_j=P$. Complete these to Kraus operators of a channel
by choosing an orthonormal basis $\{e_i\}$ of $\ker H$, a unit output
vector $\ket0$, and adding $S_i=\ket0\bra{e_i}$. Their conjugate maps sum to
$I-P$. Since $R_j\sqrt H=L_j$ and $S_i\sqrt H=0$, the resulting channel
satisfies \eqref{eq:effect-factor}. This also proves the identity for inputs
with off-diagonal blocks connecting $\supp(H)$ and $\ker(H)$ and for singular $H$.
\end{proof}

\begin{lemma}[Disturbance of normalized positive scaling]\label[lemma]{lem:weight-one}
Let $\theta$ be a state, $H\succeq0$, and $\Tr H\theta=1$.
Then $\theta_H=\sqrt H\theta\sqrt H$ is a state and, with
$d_\theta=\Tr\theta(I-H)_+$,
\begin{equation}
 \Delta(\theta_H,\theta)\le\sqrt{2d_\theta}.
 \label{eq:weight-one}
\end{equation}
\end{lemma}
\begin{proof}
Cyclicity gives $\Tr\theta_H=1$. Purify $\theta$ to a unit vector
$\Theta$. The vector $\Theta_H=(\sqrt H\otimes I)\Theta$ is also unit vector and has overlap $\langle\Theta,\Theta_H\rangle=\Tr\theta\sqrt H$.
For a scalar $h\ge0$, $\sqrt h\ge1-(1-h)_+$: when $h\le1$ this is
$\sqrt h\ge h$, and when $h\ge1$ it is $\sqrt h\ge1$.
Applying this in an eigenbasis of $H$ gives
$\sqrt H\succeq I-(I-H)_+$. Taking the trace against $\theta$ yields
$\Tr\theta\sqrt H\ge1-d_\theta$. Also $0\le d_\theta\le1$ because
$0\preceq(I-H)_+\preceq I$, and the overlap of two unit vectors is at
most one. Thus the overlap is real and nonnegative and lies in $[1-d_\theta,1]$.
The pure-state distance formula and partial-trace data processing imply
\[
 \Delta(\theta_H,\theta)
 \le\sqrt{1-(\Tr\theta\sqrt H)^2}
 \le\sqrt{1-(1-d_\theta)^2}\le\sqrt{2d_\theta}.
\]
\end{proof}

\begin{theorem}[Two-state mean-one comparison]\label[theorem]{thm:weighting}
Let $\Phi$ be a CP map and $H=\Phi^*(I)$. Suppose states $\rho,\sigma$ satisfy
$\Delta(\rho,\sigma)\le\alpha$ and
$\Tr H\rho=\Tr H\sigma=1$. Then
\begin{equation}
 \Delta(\Phi(\rho),\Phi(\sigma))
 \le\sqrt{2(d_\sigma+\alpha)}+\alpha+\sqrt{2d_\sigma},
 \quad d_\sigma=\Tr\sigma(I-H)_+.
 \label{eq:weight-two}
\end{equation}
\end{theorem}
\begin{proof}
The variational bound for the bounded test $(I-H)_+$ gives
$d_\rho\le d_\sigma+\alpha$. Both $\rho_H,\sigma_H$ are normalized.
Compare through the unscaled states and use \cref{lem:weight-one}:
\[
 \Delta(\rho_H,\sigma_H)
 \le\Delta(\rho_H,\rho)+\Delta(\rho,\sigma)+\Delta(\sigma,\sigma_H)
 \le\sqrt{2(d_\sigma+\alpha)}+\alpha+\sqrt{2d_\sigma}.
\]
Use the same channel $\Psi$ from \cref{lem:effect-factor} at both endpoints
and apply data processing to obtain \eqref{eq:weight-two}.
\end{proof}
The bounded deficit $(I-H)_+$, rather than $\|H\|_\infty$, controls the
transfer of input error. Both output traces must equal one. For classical
diagonal states the sharper identity $\Delta(Pw,P)=\E_P(1-w)_+$ is
available; Appendix~\ref{app:classical} uses it in place of the quantum
square-root disturbance bound.

\subsection{A noncommuting random-fiber calculation}
\begin{lemma}[Permutation-fiber moments]\label[lemma]{lem:moments}
Let positive matrices $\{Q_c:c\in\F\}$ obey
$\sum_cQ_c=I$ and $Q_c\preceq\kappa I$. For a fixed message $m$ and
$Z\leftarrow U_{\cK}$, define
\[
 H_Z=q\sum_{c:f_Z(c)=m}Q_c,\qquad \gamma=\frac{N-q}{N-1}.
\]
Then
\begin{align}
 \E H_Z&=I,\label{eq:first-moment}\\
 \E H_Z^2&=\gamma I+(q-\gamma)\sum_cQ_c^2,\label{eq:second-exact}\\
 \E(H_Z-I)^2&\preceq q\kappa I.\label{eq:second-bound}
\end{align}
For every normalized $\theta$ independent of the ideal key,
\begin{equation}
 \E_Z\Tr\theta(I-H_Z)_+\le\frac12\sqrt{q\kappa}.
 \label{eq:deficit-bound}
\end{equation}
\end{lemma}
\begin{proof}
Let $I_c(z)=\ind[f_z(c)=m]$. The one-point probability in
\cref{lem:fibers} gives $\E H_Z=q\sum_c(1/q)Q_c=I$.
Expand the square with \emph{ordered} pairs:
\[
 H_z^2=q^2\sum_{c,c'}I_c(z)I_{c'}(z)Q_cQ_{c'}.
\]
On the diagonal $\E I_c^2=1/q$, giving coefficient $q$.
Off the diagonal \eqref{eq:twomembership} gives coefficient
$q^2T(T-1)/(N(N-1))=(N-q)/(N-1)=\gamma$.
Therefore
\[
 \E H_Z^2=q\sum_cQ_c^2+\gamma\sum_{c\ne c'}Q_cQ_{c'}.
\]
Distributivity implies
$\sum_{c,c'}Q_cQ_{c'}=(\sum_cQ_c)(\sum_{c'}Q_{c'})=I$,
without interchanging any products. Subtract the diagonal to prove
\eqref{eq:second-exact}. 

Since $N=qT\ge q\ge2$, $0\le\gamma\le1$ and
$0\le q-\gamma\le q$. The eigenvalues of $Q_c$ lie in $[0,\kappa]$,
so $Q_c^2\preceq\kappa Q_c$ by spectral calculus. Consequently
\[
 \E(H_Z-I)^2=(\gamma-1)I+(q-\gamma)\sum_cQ_c^2
 \preceq(\gamma-1)I+(q-\gamma)\kappa I\preceq q\kappa I.
\]
All scalar coefficients used to multiply positive-order inequalities are
nonnegative. This proves \eqref{eq:second-bound}.

For any Hermitian $A$, $A_+=(|A|+A)/2$. The mean of $I-H_Z$ is zero,
so the signed term cancels after expectation:
\[
 \E\Tr\theta(I-H_Z)_+=\tfrac12\E\Tr\theta|I-H_Z|.
\]
By the state-scaled inequality of \cref{lem:tools} and then scalar
Cauchy--Schwarz,
\[
 \tfrac12\E\Tr\theta|I-H_Z|
 \le\tfrac12\E\sqrt{\Tr\theta(I-H_Z)^2}
 \le\tfrac12\sqrt{\E\Tr\theta(I-H_Z)^2}
 \le\tfrac12\sqrt{q\kappa}.
\]
The last step uses $\Tr\theta=1$.
\end{proof}
The controlled version can be written on the full side state without
conditioning on $\Lambda$. For $\lambda=(y,y',\sflag,\overline z')$ with $\sflag\in\{0,1\}$, put
\[
 \widehat Q_c=\sum_\lambda\ket\lambda\bra\lambda\otimes Q_c^{y,y'},
 \qquad
 \widehat H_z=q\sum_{c:f_z(c)=m}\widehat Q_c.
\]
These matrices obey $\sum_c\widehat Q_c=I_{\Lambda E_2}$ and
$\widehat Q_c\preceq\kappa I_{\Lambda E_2}$. Thus Lemma~\ref{lem:moments}, applied
directly to $\mu_{\Lambda E_2}$, gives
$\E_Z\Tr\mu_{\Lambda E_2}(I-\widehat H_Z)_+\le\tfrac12\sqrt{q\kappa}$.
The whole side state remains intact; only $Z$ is averaged independently.

With the continuation bound \eqref{eq:Q},
\begin{equation}
 q\kappa=2^\ell\frac{2^{d-s}}{\zeta}
          =\frac{2^{d-r}}{\zeta}.
 \label{eq:surplus}
\end{equation}
This cancellation is why padding length $r>d$ is enough: the ciphertext
has $s=\ell+r$ bits, the message selection costs $\ell$, and only the second source length costs $d$.

\section{Worst-case security and the message-free simulator}
\label{sec:security}
We prove \cref{thm:main}, using the continuation and scaling lemmas
to compare each fixed-message output with one message-free simulator.

\subsection{Compute the symbol before forgetting the key}
Fix an original attack and its smoothed version $\cA_\zeta$. The states
$\mu,\sigma$ in Eq.~\eqref{eq:ideal}, the continuation, and the hybrid register
$\nu$ in Eq.~\eqref{eq:nu} are all fixed without choosing a message.
For a register value $\lambda=(y,y',\sflag,\overline z')$, define
\begin{equation}
 \Sym(z,\lambda,c,c')=
 \begin{cases}
 \same,&\sflag=1,\ c'=c,\\
 f_z(c'),&\sflag=1,\ c'\ne c,\\
 f_{\overline z'}(c'),&\sflag=0.
 \end{cases}
 \label{eq:j-symbol}
\end{equation}
On inconsistent labels, namely $\sflag=0,\overline z'=\same$ or
$\sflag=1,\overline z'\ne\same$, set $\Sym(z,\lambda,c,c')=\perp$ instead. This makes $\Sym$ a total classical function on its formal alphabet.

\begin{lemma}[Exact copy representation]\label[lemma]{lem:copy}
For a real encoding of $m$ under $\cA_\zeta$, let $J_m$ be the symbol
computed by \eqref{eq:j-symbol}. Then the actual decoded output is
$\Copy(m,J_m)$ on every execution.
\end{lemma}
\begin{proof}
If $S=0$, the masked value is the actual tampered seed, and the last line
of \eqref{eq:j-symbol} is exactly the decoder output. If $S=1$, the
original and tampered $\nmext$ inputs coincide, so $Z'=Z$. When also $C'=C$,
correctness gives $m$, which is the copy interpretation of $\same$.
Otherwise $f_Z(C')$ is the actual decoder output and is explicitly the
second line.
\end{proof}
The second line of $\Sym$ uses the original key, so the following map computes
the symbol before discarding $Z$.

For a fixed message define $\chi_m(z,c)=\ind[f_z(c)=m]$ and let $\Pi_m$
project onto this event on the original $Z,C$. Let $\mathcal R^+$ be
continuation retaining $Z$, $\mathcal T_m(A)=q\Pi_m A\Pi_m$, and
$\mathcal J$ the CPTP map computing $\Sym$ and discarding other outputs. Then
$\Phi_m=\mathcal J\circ\mathcal T_m\circ\mathcal R^+$.
The first and last maps are CPTP; $\mathcal T_m$ is a CP map but not generally
trace-nonincreasing. For an input matrix $A$ on $Z\Lambda E_2$, define
its diagonal block on $E_2$ by
\[
 A_{z\lambda,z\lambda}:=
 (\bra{z,\lambda}\otimes I_{E_2})A
 (\ket{z,\lambda}\otimes I_{E_2}).
\]
With the classical labels dephased, the full formula for $\Phi_m$ is
\begin{equation}
 \Phi_m(A)=q\sum_{z,\lambda,c,c'}\chi_m(z,c)
  \Tr\!\left[\mathcal R^{y,y'}_{c,c'}(A_{z\lambda,z\lambda})\right]
  \ket{\Sym(z,\lambda,c,c')}\bra{\Sym(z,\lambda,c,c')}.
 \label{eq:Phi}
\end{equation}
Every step is CP, including positive scalar multiplication, so this defines
a CP map on all matrices. It need not be trace-preserving on arbitrary
inputs. Its output-scale operator $H=\Phi_m^*(I)$ is block diagonal in
$z,\lambda$, with blocks
\begin{equation}
 H_z^\lambda=q\sum_{c:f_z(c)=m}Q_c^{y,y'}.
 \label{eq:H-block}
\end{equation}
Indeed, the output trace on a positive input supported on one classical
block is
\[
 q\sum_{c,c'}\chi_m(z,c)\Tr\mathcal R^{y,y'}_{c,c'}(A_{E_2}),
\]
which equals $\Tr H_z^\lambda A_{E_2}$. Computing the key-dependent symbol
does not affect this trace. Also $0\preceq H_z^\lambda\preceq qI$,
as the selected $Q_c$ form a sub-sum of a family totaling $I$.

The selector acts only on the original $Z,C$, after continuation has recovered
$C$ while retaining $Z$. BBJ is applied before either operation. The factor
$q$ normalizes the two complete states below, not each transcript.
Figure~\ref{fig:conditioning-bbj} separates the CP selection from its
normalization; Figure~\ref{fig:comparison-bbj} shows the two comparisons.

\subsection{The two exact normalizations}
\begin{lemma}[Actual and ideal message-selected states]\label[lemma]{lem:normalized}
The map in Eq.~\eqref{eq:Phi} satisfies
\begin{equation}
 \Phi_m(\mu)=P_{J_m},\qquad
 \Tr\Phi_m(\mu)=\Tr\Phi_m(\sigma)=1.
 \label{eq:two-norms}
\end{equation}
Moreover, in the ideal selected experiment, the marginal of the refined
register before computing $\Sym$ and discarding $Z$ is exactly $\nu$ after
$Z$ is traced out, for every $m$.
\end{lemma}
\begin{proof}
By Eq.~\eqref{eq:composition}, continuing each summand of Eq.~\eqref{eq:mu}
gives the actual uniform-ciphertext hybrid with its original key retained.
The message test and multiplier $q$ therefore impose exactly its event
$M=m$, by Lemma~\ref{lem:real-hybrid}. For a fixed source pair, its input
coefficient becomes
$q\,2^{-n-d}N^{-1}\chi_m(z(x,y),c)=2^{-n-d}T^{-1}\chi_m(z(x,y),c)$,
which is exactly \eqref{eq:real-law}. There are $T$ surviving ciphertexts
per source pair. The attack CP maps sum to trace one on each input, so
the total selected trace is one. Computing $\Sym$ gives $J_m$ exactly.

For the ideal input, write the full continuation output as
\[
 \nu_{\Lambda CC'E'_2}=
 \sum_{\lambda,c,c'}\ket{\lambda,c,c'}\bra{\lambda,c,c'}
       \otimes\tau_{E'_2}^{\lambda,c,c'},\qquad
 \tau_{E'_2}^{\lambda,c,c'}=
       \mathcal R^{y,y'}_{c,c'}(\vartheta_{E_2}^\lambda),
\]
where $\mu_{\Lambda E_2}=\sum_\lambda\ket\lambda\bra\lambda
\otimes\vartheta_{E_2}^\lambda$. All blocks retain their probability scales.
Before message selection the joint state is $U_{\cK}\otimes\nu$.
Let $\Pi_m$ be the classical projector onto $f_Z(C)=m$; the selected state is
\[
 \widetilde\sigma_m=q\Pi_m(U_{\cK}\otimes\nu)\Pi_m.
\]
Tracing out $Z$, each $(\lambda,c,c')$ block is multiplied by
$q\sum_z u_{\cK}(z)\chi_m(z,c)=1$, by \eqref{eq:onepoint}. Consequently
$\Tr_Z\widetilde\sigma_m=\nu$ as a full state, including $E'_2$.
Its trace is one, and computing $\Sym$ preserves this trace. This proves the
ideal normalization without normalizing any individual transcript block.
\end{proof}
These are global trace identities, not separate normalizations of each
$(z,\lambda)$ block. Set $\Omega_m:=\Phi_m(\sigma)$; its remaining
message dependence is the finite-permutation correction analyzed below.

\begin{lemma}[Fixed-message output comparison]\label[lemma]{lem:comparison}
For every $m$, with the same states $\mu,\sigma$,
\begin{equation}
 \Delta\bigl(P_{J_m},\Omega_m\bigr)\le\beta.
 \label{eq:comparison}
\end{equation}
\end{lemma}
\begin{proof}
The input distance is at most $\alpha$ by \eqref{eq:ideal}.
Using the full-side operators $\widehat H_z$ from
Section~\ref{sec:weighting}, the block form \eqref{eq:H-block} is
$H=\sum_z\ket z\bra z\otimes\widehat H_z$. Since
$\sigma=U_{\cK}\otimes\mu_{\Lambda E_2}$, blockwise spectral calculus gives
\[
 d_\sigma=\Tr\sigma(I-H)_+
     =\E_{Z\leftarrow U_{\cK}}
        \Tr\mu_{\Lambda E_2}(I-\widehat H_Z)_+.
\]
Apply the controlled moment bound directly to the normalized joint state
$\mu_{\Lambda E_2}$. Equations~\eqref{eq:Q}, \eqref{eq:deficit-bound}, and
\eqref{eq:surplus} give
\begin{equation}
 d_\sigma\le\frac12\sqrt{\frac{2^{d-r}}{\zeta}}.
 \label{eq:actual-deficit}
\end{equation}
Apply \cref{thm:weighting} using both normalizations in
\cref{lem:normalized}. Since $\sqrt{u+v}\le\sqrt u+\sqrt v$ for
nonnegative $u,v$, the output distance is at most
\[
 \sqrt{2(d_\sigma+\alpha)}+\alpha+\sqrt{2d_\sigma}
 \le\alpha+\sqrt{2\alpha}+2\sqrt{2d_\sigma}
 \le\beta.
\]
The identity $\Phi_m(\mu)=P_{J_m}$ finishes the claim.
\end{proof}
The same intermediate states and bound apply to every $m$; no union bound
or $q$-fold amplification of the $\nmext$ error is used.

\subsection{One simulator for all messages}
Define $D_{\cA,\zeta}$ by the following procedure:
\begin{lstlisting}[mathescape=true]
Sample $X,Y,C$ independently and uniformly.
Run the original left attack and the $\zeta$-smoothed right attack.
Let $S=1$ exactly when $(X',Y')=(X,Y)$.
If $S=0$:
    output $f_{\nmext(X',Y')}(C')$.
Else if $C'=C$:
    output $\same$.
Else:
    output $\ell$ fresh independent fair bits.
\end{lstlisting}
This procedure has no message input and samples the actual smoothed
hybrid. It may retain the classical inputs it sampled to compute $S$.
It runs the specified attack and shared state, not the algebraic coarse
channel. There is no inverse square root or postselection to implement.

\begin{lemma}[Ideal simulation]\label[lemma]{lem:simulator}
For every $m$, $\Delta(\Omega_m,D_{\cA,\zeta})\le\epsp$.
\end{lemma}
\begin{proof}
Let $P^{\rm hyb}_{\Lambda CC'}(\lambda,c,c')=
\Tr\tau^{\lambda,c,c'}$ be the classical distribution obtained from $\nu$ by
discarding $E'_2$. By \cref{lem:normalized}, the ideal selected joint distribution is
\[
 \Pr_{\widetilde\sigma_m}[Z=z,\Lambda=\lambda,C=c,C'=c']
 =q\,u_{\cK}(z)\chi_m(z,c)
               P^{\rm hyb}_{\Lambda CC'}(\lambda,c,c').
\]
Summing over $z$ leaves $P^{\rm hyb}_{\Lambda CC'}$. On $S=0$ the symbol uses only $\overline Z',C'$, and on $S=1,C'=C$ it is $\same$. On $S=1,C'\ne C$, using \cref{lem:fibers} gives the same distribution $p_m$ from \eqref{eq:prefix-law} for every supported register. The simulator substitutes $U_\ell$ for that
distribution. Quantum blocks factor out of this classical key sum.

For an explicit calculation, let $w_{\rm same}$ and $w_{\rm change}$
be the hybrid probabilities of $S=1,C'=C$ and $S=1,C'\ne C$.
Let $\gamma_u$ be the hybrid probability of $S=0$ and tampered decoding
$u$. These nonnegative numbers sum to one. Then, writing $\delta_j$
for a point mass,
\begin{align}
 \Omega_m&=w_{\rm same}\delta_{\same}
               +w_{\rm change}p_m+\sum_u\gamma_u\delta_u,\notag\\
 D_{\cA,\zeta}&=w_{\rm same}\delta_{\same}
               +w_{\rm change}U_\ell+\sum_u\gamma_u\delta_u.
 \label{eq:sim-mixture}
\end{align}
Their distance is exactly $w_{\rm change}\epsp\le\epsp$.
All scaling factors come from the message-free hybrid, not from a choice depending
on $m$.
\end{proof}
The $\epsp$ term accounts for the slightly deficient mass at $m$ in
$p_m$; it is not reinterpreted as a negative $\same$ probability.

\begin{proof}[Proof of \cref{thm:main}]
Correctness, complexity, and rate follow from \cref{lem:correctness}.
Fix an arbitrary attack and $\zeta$, and use the simulator above. For
any $m$, \cref{lem:comparison,lem:simulator} give
$\Delta(P_{J_m},D_{\cA,\zeta})\le\beta+\epsp$.
Apply the same deterministic map $\Copy(m,\cdot)$ and use the exact
identity of \cref{lem:copy}. Thus the trace distance between the smoothed
real decoded-output distribution and $\Copy(m,D_{\cA,\zeta})$ is at most
$\beta+\epsp$. Comparing the original attack with the smoothed attack
adds at most $\zeta$ by \cref{lem:fallback}, proving
\eqref{eq:main-error}. The simulator was defined before $m$, so the same
distribution works for all messages.

Every CP map identity holds on arbitrary entangled inputs, and the
operator inequalities and normalized trace estimates contain no local
dimension parameter. Thus the error is uniform over the stipulated
entanglement dimensions. The quantitative simplification
\eqref{eq:main-asymptotic} is verified in \cref{sec:parameters} below.
\end{proof}

\section{Parameters, rate, and scope}
\label{sec:parameters}

\subsection{Finite error and the padding surplus}
Expanding \eqref{eq:main-error} gives
\begin{equation}
 \epsn\le\zeta+\alpha+\sqrt{2\alpha}
     +2\left(\frac{2^{-g}}\zeta\right)^{1/4}
     +\frac{q-1}{q(N-1)},
 \quad \alpha=\eta+2^{-s},\quad g=r-d.
 \label{eq:error-expanded}
\end{equation}
For $g>0$, let $k_{\rm sm}=\lceil g/5\rceil$ and $\zeta=2^{-k_{\rm sm}}$.
This is a directly samplable proof parameter: test whether $k_{\rm sm}$ fair bits
are all zero. The rounding inequalities give
\[
 \zeta\le2^{-g/5},\qquad
 (2^{-g}/\zeta)^{1/4}=2^{(-g+k_{\rm sm})/4}
             \le2^{1/4}2^{-g/5}.
\]
For $\alpha\le1$, $\alpha\le\sqrt\alpha$ and
$\sqrt\alpha\le\sqrt\eta+2^{-s/2}$.
Also $\epsp\le1/(N-1)\le2^{1-s}$. Since $0<g<s$, both small seed
terms are bounded by constant multiples of $2^{-g/5}$. Substitution proves
\eqref{eq:main-asymptotic}. This does not assert that the exponent $1/5$
in the \emph{error} is optimal. It arises from two square-root losses and from balancing the smoothing
error with the fourth-root term; it is distinct from the limiting code rate $1/5$.

\subsection{Instantiating the BBJ dimensions}
Fix a sufficiently small allowed constant $0<\delta<1/10$ and use
\eqref{eq:bbj-params}. Choose
\begin{equation}
 s=\lfloor s_0/2\rfloor,\qquad r=2d,\qquad\ell=s-r.
 \label{eq:chosen-lengths}
\end{equation}
For all sufficiently large $n$, $\ell\ge1$. At most one inner-output bit
is discarded. There is no requirement that $\ell$ divide $r$: the field has degree $s$,
and the message is simply a prefix in its binary representation.

Write $s_0=(1/2-\delta)n+e_n$ and $d=\delta n+f_n$ with bounded
rounding terms. For $j_n\in\{0,1\}$, $2s=s_0-j_n$. Therefore
\begin{align}
 s&=(1/4-\delta/2)n+(e_n-j_n)/2,\notag\\
 \ell&=\frac{1-10\delta}{4}n+O(1),\qquad
 g=d=\delta n+O(1),\label{eq:message-gap}\\
 L_{\rm tot}=n+d+s&=\frac{5+2\delta}{4}n+O(1).
 \label{eq:total-length}
\end{align}
Only $X,Y,C$ are transmitted; $Z$ is recomputed and $R$ is already
inside $C$. Dividing the message coefficient by the positive total-length
coefficient gives
\begin{equation}
 \lim_{n\to\infty}\frac{\ell}{L_{\rm tot}}
     =\frac{1-10\delta}{5+2\delta}<\frac15.
 \label{eq:rate-final}
\end{equation}
Its difference from $1/5$ is $52\delta/[5(5+2\delta)]$, tending to zero
as one chooses smaller allowed fixed constants. We choose $\delta$ first, with the
required margin, and only then take $n$ large. 

\begin{corollary}[Constant-rate family]\label[corollary]{cor:rate}
For every sufficiently small fixed $\xi>0$, suitable allowed BBJ parameters
give rate at least $1/5-\xi$ for sufficiently large $n$, perfect
correctness, polynomial-time online encoding and decoding, and
$2^{-n^{\Omega(1)}}$ error against every entangled two-split attack.
\end{corollary}
\begin{proof}
Choose fixed $\delta>0$ with the limit in \eqref{eq:rate-final} strictly
larger than $1/5-\xi$. Bounded rounding affects rate by $o(1)$.
The gap in \eqref{eq:message-gap} is linear. If the $\nmext$ error is
at most $2^{-c n^{a_{\rm err}}}$ for fixed $c,a_{\rm err}>0$, its square root is $2^{-(c/2)n^{a_{\rm err}}}$, while the padding error is $2^{-\Omega(n)}$.
Their sum is stretched-exponentially small, and hence negligible.
Indeed, each exponent eventually exceeds any fixed multiple of
$\log_2 n$; the fixed prefactor in \eqref{eq:main-asymptotic} can be
absorbed by reducing the exponential constant. Complexity and correctness
are \cref{lem:correctness}.
\end{proof}

\subsection{Equal-length components}
\label{sec:balanced}
The rate in \eqref{eq:rate-final} counts unequal component lengths. For
comparison with the balanced model, padding the shorter component gives
the following separate guarantee, without changing the error bound.

\begin{corollary}[Balanced-share code]\label[corollary]{cor:balanced}
For the fixed-$\delta$ family in \eqref{eq:chosen-lengths}, and sufficiently
large $n$, there is a perfectly correct entangled two-split NMC whose two
components each have length $n$, with the same security error and limiting
rate $(1-10\delta)/8$. Consequently, balanced rates can approach $1/8$
from below through sufficiently small allowed fixed $\delta>0$.
\end{corollary}
\begin{proof}
The original right length is
$d+s=(1/4+\delta/2)n+O(1)<n$ for large $n$. Set $b_n=n-(d+s)$.
Encode using \cref{con:code}, then append $b_n$ zero bits to the right
component. The new decoder discards these final bits and applies the old
decoder. Honest decoding is therefore unchanged.

Fix any entangled attack on the balanced code. Form an attack on the
original code by leaving its left procedure unchanged and making its right
procedure append those zeros, run the balanced right attack, and discard
the last $b_n$ output bits. These are local operations on the right alone,
using the same pre-shared state and no communication. For each message,
the original code under this induced attack has exactly the same decoded
output distribution as the balanced code under the given attack. The simulator
provided by \cref{thm:main} for the induced attack is consequently one
simulator for every message of the balanced code, with the same error.

Only storage length changes: the total is $2n$, while
$\ell=(1-10\delta)n/4+O(1)$. Dividing gives limiting rate
$(1-10\delta)/8$. Polynomial-time encoding and decoding are preserved.
\end{proof}

Our unbalanced rate, balanced rate, and error exponent are distinct
parameters. In particular, the stated BBJ instantiation does not improve
Li's exponential classical error~\cite{Li23}.

\section*{Acknowledgements}
\addcontentsline{toc}{section}{Acknowledgements}
\paragraph{Funding.}
The author thanks the National Research Foundation, Singapore (NRF), for funding the Senior Research Fellow position through the NRF--NUS Postdoctoral Award.

\paragraph{AI disclosure.}
 The author used generative-AI systems from OpenAI (ChatGPT) and Anthropic (Claude) during both research development and manuscript preparation. The author supplied multiple candidate non-malleable code constructions and used the systems interactively to explore possible security arguments. In one such interaction, an AI system suggested a proof strategy for an author-proposed non-malleable code construction yielding rate approximately \(1/14\). After independently studying this argument, the author recognized that the underlying strategy could be adapted to a modified version of the BBJ permutation construction, leading to the rate-\(1/5\) construction studied in this paper. AI systems were subsequently used for further proof discussions and for drafting and revising early versions of the manuscript. The author independently checked the proofs and calculations and substantially edited the resulting text through multiple revisions. The author takes full responsibility for the correctness, originality, and presentation of the final paper.

\label{end:main}

\Needspace{8\baselineskip}
\phantomsection\addcontentsline{toc}{section}{References}
\bibliography{references}
\bibliographystyle{alpha}   
\clearpage
\appendix
\renewcommand{\theHsection}{appendix.\Alph{section}}
\renewcommand{\theHsubsection}{\theHsection.\arabic{subsection}}
\section{A complete classical proof of the rate-one-fifth construction}
\label{app:classical}
This appendix proves the security of the construction against classical split-state
tampering without using the quantum scaling theorem, quantum channels,
or quantum conditional states. We restate the construction and every
finite-field and probabilistic fact needed in its proof. The only external
cryptographic input is an explicitly specified classical randomness encoder.
The classical specialization of BBJ supplies that input; its separate
extractor construction is not reproved here.

The classical proof is quantitatively stronger than simply restricting the
quantum theorem to classical attacks. It uses exact posterior distributions
and averages their sizes instead of adding the quantum proof's smoothing.
Its final error is
\[
 2(\eta+2^{-s})+2^{-(r-d)/2}
       +\frac{2^\ell-1}{2^\ell(2^s-1)}.
\]
The message and codeword lengths are the same as in the main construction.

\subsection{Probability conventions and the security goal}
For finite probability distributions $P,Q$ on one alphabet,
\[
 \Delta(P,Q)=\frac12\sum_x|P(x)-Q(x)|.
\]
If $a_x=P(x)-Q(x)$, then $\sum_xa_x=0$. The total positive and negative
parts have equal magnitude, each equal to $\Delta(P,Q)$. Therefore,
for every function $b:\mathcal X\to[0,1]$,
\begin{equation}
 |\E_Pb-\E_Qb|\le\Delta(P,Q).
 \label{eq:cl-bounded}
\end{equation}
Indeed, the largest possible positive sum $\sum_xb_xa_x$ is at most
$\sum_{a_x>0}a_x$, and the negative sum is treated the same way.

A classical channel is a conditional distribution
$P_{B\mid A}(b\mid a)=\Pr[B=b\mid A=a]$, with nonnegative probabilities
summing to one over $b$. Applying this same channel to input distributions $P_A,Q_A$
gives output distributions
$P_B(b)=\sum_a P_A(a)P_{B\mid A}(b\mid a)$ and
$Q_B(b)=\sum_a Q_A(a)P_{B\mid A}(b\mid a)$. Hence
\[
 \Delta(P_B,Q_B)\le\frac12\sum_a|P_A(a)-Q_A(a)|
                         \sum_bP_{B\mid A}(b\mid a)
                 =\Delta(P_A,Q_A).
\]
For joint distributions, forgetting a coordinate is one such classical channel.
Appending the same independent distribution preserves distance, by factoring its
probabilities out of the absolute-value sum. Finally, if distributions with disjoint
labels have the same scaling factors, their joint distance is the
probability-scaled sum of conditional distances.

Write $(\LeftShare,\RightShare)$ for the two codeword shares, reserving $R$ for
random padding. A classical deterministic split attack maps
$(\LeftShare,\RightShare)$ to $(f(\LeftShare),g(\RightShare))$.
Randomized attacks may have arbitrarily correlated classical random tapes
independent of the message and encoder coins, but may not communicate after
receiving the shares. Non-malleability requires one distribution $D_{f,g}$ on
$\bits^\ell\cup\{\same,\perp\}$ for all messages, with
\[
 \Delta\bigl(\Dec(f(\LeftShare),g(\RightShare)),\Copy(m,D_{f,g})\bigr)\le\varepsilon,
 \qquad (\LeftShare,\RightShare)\leftarrow\Enc(m).
\]
The definition of $\Copy$ is $\Copy(m,\same)=m$ and
$\Copy(m,u)=u$ for other symbols. For a distribution $D$,
$\Copy(m,D)$ is the law of $\Copy(m,J)$ for $J\sim D$.
We first analyze deterministic functions,
then average the corresponding joint distributions over shared random tapes.
For example, $(\LeftShare',\RightShare')=(f(\LeftShare),g(\RightShare))$ has conditional
distribution
\[
 P_{\LeftShare'\RightShare'\mid\LeftShare\RightShare}
 (l_{\rm sh}',r_{\rm sh}'\mid l_{\rm sh},r_{\rm sh})
 =\ind[l_{\rm sh}'=f(l_{\rm sh})]\ind[r_{\rm sh}'=g(r_{\rm sh})].
\] This factorization is used only for the
classical deterministic case. Superscript ${\rm hyb}$ below denotes the
uniform-ciphertext hybrid.

\subsection{The external classical non-malleable randomness encoder}
We use $\nmext$ for the truncated $2s$-bit primitive, as in the main text;
it is distinct from the full $s_0$-bit primitive $\nmext_0$.
Let
\[
 \nmext:\bits^n\times\bits^d\longrightarrow\bits^{2s}
\]
be polynomial-time computable. The classical premise is as follows.
For independent uniform $X,Y$, let $Z=\nmext(X,Y)$.
Apply any pair of local classical randomized operations to obtain $X',Y'$,
and define
\[
 S=\ind[(X',Y')=(X,Y)],\qquad
 \overline Z'=\begin{cases}\same,&S=1,\\\nmext(X',Y'),&S=0.\end{cases}
\]
Writing $\Lambda=(Y,Y',S,\overline Z')$, assume
\begin{equation}
 \Delta\bigl(P_{Z\Lambda},U_{2s}\otimes P_\Lambda\bigr)\le\eta.
 \label{eq:cl-inner}
\end{equation}
Only the stated register is exposed, not the full original $X$.

This is an assumption about a classical function and classical distributions.
It suffices for the proof below, regardless of how it is established.
The right-side scaled guarantee of BBJ~\cite[Lemma 3]{BBJ23}, restricted
to classical local operations and with the right source and output retained,
supplies it at
\begin{equation}
 d=\delta n+O(1),\quad 2s\le(1/2-\delta)n+O(1),\quad
 \eta=2^{-n^{\Omega(1)}}.
 \label{eq:cl-inner-params}
\end{equation}
A scaled same-input/changed-input statement gives
\eqref{eq:cl-inner} directly: label the two cases, mask the equal output
in the first, and sum the absolute differences in the two disjoint blocks.
A zero-probability case contributes zero. Output truncation is deterministic
processing and preserves the error. For randomized local maps, retain a classical random-tape variable $\Omega$ with its source-independent distribution. Apply the deterministic bound to
the conditional distributions given $\Omega$, and average with $P_\Omega$.
Forgetting this label cannot increase total variation. 

\subsection{The non-malleable code and the ideal permutation distribution}
Put $s=\ell+r$, $N=2^s$, $q=2^\ell$, and $T=2^r$, so $N=qT$.
Fix a binary representation of $\F_N$. For $z=(a,b)\in\F_N^2$ define
$a^\star=a$ if $a\ne0$, otherwise $a^\star=1$, and
\[
 \pi_z(u)=a^\star u+b,\qquad
 f_z(c)=\operatorname{pref}_\ell\bigl((a^\star)^{-1}(c-b)\bigr).
\]
To encode $m$, independently choose $X\leftarrow U_n$, $Y\leftarrow U_d$,
$R\leftarrow U_r$ and let
\begin{equation}
 Z=\nmext(X,Y),\quad C=\pi_Z(m\|R),\quad
 (\LeftShare,\RightShare)=(X,(Y,C)).
 \label{eq:cl-code}
\end{equation}
Decode the tampered strings by $f_{\nmext(X',Y')}(C')$.
The slope is always nonzero, so inverting honest $C$ gives $m\|R$ for
all source and padding choices. Thus correctness is perfect. The algorithms
use polynomial-time field operations, polynomial-time computable $\nmext$ and transmit $n+d+s$ bits. The rate is $\ell/(n+d+s)$.

Let $\cK=\F_N^\times\times\F_N$. The exact distance from uniform
raw seeds to $U_{\cK}$ is $1/N$: the excluded event $a=0$ has precisely
that probability. More explicitly, uniform raw seeds give mass $1/N^2$
to each pair, whereas $U_{\cK}$ gives $1/(N(N-1))$ to pairs with nonzero
slope. The total positive difference on allowed pairs is $1/N$, equalling the missing mass on excluded pairs. We write
$u_{\cK}(z)=\ind[z\in\cK]/(N(N-1))$ for the probability mass function
of $U_{\cK}$ on the raw seed alphabet $\F_N^2$.

\begin{lemma}[Classical affine counting]\label[lemma]{lem:cl-affine}
For every raw $z,m$, exactly $T$ ciphertexts satisfy $f_z(c)=m$, and
$\pi_z(m\|U_r)$ samples them uniformly. For fixed $c\ne c'$ and an
ideal uniform key,
\begin{align}
 \Pr[f_Z(c)=m]&=1/q,\notag\\
 \Pr[f_Z(c)=m,f_Z(c')=m]&=T(T-1)/(N(N-1)).
 \label{eq:cl-membership}
\end{align}
The conditional distribution of $f_Z(c')$ given $f_Z(c)=m$ is $p_m$ with
\begin{equation}
 p_m(u)=\begin{cases}(T-1)/(N-1),&u=m,\\T/(N-1),&u\ne m,
 \end{cases}\qquad
 \Delta(p_m,U_\ell)=\frac{q-1}{q(N-1)}=\epsp.
 \label{eq:cl-prefix}
\end{equation}
\end{lemma}
\begin{proof}
For each raw key, the $T$ padded plaintexts map bijectively to the desired
ciphertexts, proving the first claim. For any distinct plaintexts $u,u'$
and distinct ciphertexts $c,c'$, the system $au+b=c$, $au'+b=c'$ has
unique solution $a=(c-c')/(u-u')\ne0$, $b=c-au$. Thus the inverse pair
under an ideal uniform key is uniform over all $N(N-1)$ distinct ordered
pairs. The first inverse has prefix $m$ with probability $T/N$.
There are $T(T-1)$ inverse pairs with both prefixes $m$, and $T^2$ pairs
with first prefix $m$ and any fixed different second prefix. Divide these
counts by $N(N-1)$ and then by $T/N$ for the conditional distribution.
Its deficit at $m$ from $1/q$ is $(q-1)/(q(N-1))$, and each of the
$q-1$ other prefixes has excess $1/(q(N-1))$. Their equal total magnitudes
give the displayed variation distance.
\end{proof}

\subsection{An exact classical posterior factorization}
Consider deterministic outer tampering
\[
 X'=f(X),\qquad Y'=g_1(Y,C),\qquad C'=g_2(Y,C).
\]
Its right conditional distribution is
$P_{Y'C'\mid YC}(y',c'\mid y,c)=
\ind[y'=g_1(y,c)]\ind[c'=g_2(y,c)]$.
In the message-free hybrid, $P^{\rm hyb}_{XYC}=U_n\times U_d\times U_s$.
Write $z(x,y)=\nmext(x,y)$ and $p(x,y)=2^{-n-d}$.
Its hybrid message $M=f_Z(C)$ has probability $1/q$ for every fixed
$x,y$, by the first part of \cref{lem:cl-affine}. Therefore the hybrid
conditioned on $M=m$ has distribution
\begin{equation}
 P[x,y,c\mid M=m]=p(x,y)T^{-1}\ind[f_{z(x,y)}(c)=m],
 \label{eq:cl-real-law}
\end{equation}
exactly the real encoding distribution. Applying the fixed tampering maps
to equal input distributions preserves this equality.

Define $A_{y,y'}=\{c:g_1(y,c)=y'\}$ and $h_{y,y'}=|A_{y,y'}|$.
The sets partition the ciphertext alphabet for each $y$, so
$\sum_{y'}h_{y,y'}=N$. Marginalizing the hybrid joint distribution gives
\[
 P^{\rm hyb}_{Y'\mid Y}(y'\mid y)=\frac{h_{y,y'}}N,\qquad
 P^{\rm hyb}_{C\mid Y,Y'}(c\mid y,y')=
       \frac{\ind[c\in A_{y,y'}]}{h_{y,y'}}
\]
when $h_{y,y'}>0$. For an empty set complete the latter conditional
distribution arbitrarily; its conditioning event has probability zero.
The identity of conditional probabilities
\begin{equation}
 P^{\rm hyb}_{Y'\mid Y}(y'\mid y)
 P^{\rm hyb}_{C\mid Y,Y'}(c\mid y,y')
       =N^{-1}\ind[g_1(y,c)=y']
 \label{eq:cl-posterior}
\end{equation}
therefore holds for every $y,y',c$, including empty sets. It reconstructs
the full hybrid distribution by sampling $X,Y$, then $Y'$ from
$P^{\rm hyb}_{Y'\mid Y}$, then $C$ from $P^{\rm hyb}_{C\mid Y,Y'}$, and setting
$C'=g_2(Y,C)$.

Before the continuation map, compute
\[
 \lambda(x,y,y')=(y,y',\sflag(x,y,y'),\overline z'(x,y,y')),
\]
where $\sflag(x,y,y')=\ind[f(x)=x,\ y'=y]$, and the unmasked value is
$\nmext(f(x),y')$. Let $\mu_{Z\Lambda}(z,\lambda)$ be the joint distribution of this coarse
register and $z(x,y)$. It is the output of a legitimate attack on $\nmext$: the
left operation is $f$, and the right classical channel is
$P^{\rm hyb}_{Y'\mid Y}$, with no ciphertext input. Applying
\eqref{eq:cl-inner}, and then changing a fresh raw uniform key to the ideal
distribution, yields
\begin{equation}
 \sigma_{Z\Lambda}(z,\lambda)=u_{\cK}(z)\mu_\Lambda(\lambda),\qquad
 \Delta(\mu,\sigma)\le\alpha:=\eta+N^{-1}.
 \label{eq:cl-ideal}
\end{equation}
The side labels keep their actual marginal. In both distributions their marginal on $Y,Y'$ is
\begin{equation}
 \mu_{YY'}(y,y')=2^{-d}h_{y,y'}/N.
 \label{eq:cl-yv}
\end{equation}
Indeed, the joint distribution of the coarse variables is
$P^{\rm hyb}_{YY'}(y,y')=2^{-d}P^{\rm hyb}_{Y'\mid Y}(y'\mid y)$.
Marginalizing over the other retained variables leaves this distribution unchanged.

The whole continued distribution is
\[
 P^{\rm hyb}_{Z\Lambda CC'}(z,\lambda,c,c')=
 \mu_{Z\Lambda}(z,\lambda)
 P^{\rm hyb}_{C\mid Y,Y'}(c\mid y,y')\ind[c'=g_2(y,c)].
\]
Thus, the same conditional distribution reconstructs the seed, flag, and
tampered key jointly with $C,C'$. This exact factorization follows from
\eqref{eq:cl-posterior}. Its ideal counterpart replaces only
$\mu_{Z\Lambda}$ by $\sigma_{Z\Lambda}$.

\subsection{A mean-one scaling factor lemma}
The next argument is a purely probabilistic substitute for a naive bound obtained by the largest scaling factor $q$.

\begin{lemma}[Mean-one scaling factors and a conditional output distribution]\label[lemma]{lem:cl-weight}
Let $P,Q$ be probability distributions on the same finite input alphabet
$\mathcal I$, with $\Delta(P,Q)\le\alpha$. Let $\mathcal J$ be a finite
output alphabet, and let $A(j,i)\ge0$ for $i\in\mathcal I$, $j\in\mathcal J$.
For each input value define $w(i)=\sum_{j\in\mathcal J}A(j,i)$.
The function $w:\mathcal I\to[0,\infty)$ is a \emph{scaling factor}, not necessarily
a probability distribution. Define
\[
 (Pw)(i)=P(i)w(i),\qquad (PA)(j)=\sum_{i\in\mathcal I}P(i)A(j,i),
\]
and similarly $Qw,QA$. Suppose $\E_Pw=\E_Qw=1$. Then $Pw,Qw$ are
distributions on $\mathcal I$, $PA,QA$ are distributions on $\mathcal J$, and
\begin{equation}
 \Delta(PA,QA)\le2\alpha+2 d_Q,
 \qquad  d_Q=\E_Q(1-w)_+.
 \label{eq:cl-weight}
\end{equation}
\end{lemma}
\begin{proof}
All entries are nonnegative. The mean-one assumptions give
$\sum_iP(i)w(i)=\sum_iQ(i)w(i)=1$, so $Pw,Qw$ are normalized on the
\emph{input} alphabet. This is pointwise rescaling; there is no change of
alphabet at this stage. Since $\sum_iP(i)(w(i)-1)=0$, positive and negative
signed differences have equal total mass, and hence
\[
 \Delta(Pw,P)=\tfrac12\sum_iP(i)|w(i)-1|=\E_P(1-w)_+.
\]
The same identity holds for $Q$. Because $0\le(1-w(i))_+\le1$,
\eqref{eq:cl-bounded} gives $\E_P(1-w)_+\le d_Q+\alpha$. Thus
\[
 \Delta(Pw,Qw)
 \le\Delta(Pw,P)+\Delta(P,Q)+\Delta(Q,Qw)
 \le2\alpha+2 d_Q.
\]

Now define one conditional output distribution, independently of whether
its input distribution is $Pw$ or $Qw$:
\[
 P_{J\mid I}(j\mid i)=\frac{A(j,i)}{w(i)}\quad\text{when }w(i)>0.
\]
It sums to one over $j$. If $w(i)=0$, nonnegativity implies $A(j,i)=0$
for every $j$, and both $(Pw)(i)$ and $(Qw)(i)$ are zero. Complete
$P_{J\mid I}(\cdot\mid i)$ by any distribution on such inputs.
For all $i,j$ we then have $A(j,i)=w(i)P_{J\mid I}(j\mid i)$. Consequently
\[
 (PA)(j)=\sum_i(Pw)(i)P_{J\mid I}(j\mid i),\qquad
 (QA)(j)=\sum_i(Qw)(i)P_{J\mid I}(j\mid i).
\]
This is the alphabet-changing step: sample $I$ from the normalized scaled
input distribution and then sample $J$ from the \emph{same} classical channel.
In particular, $\sum_j(PA)(j)=\sum_i(Pw)(i)=1$, and likewise for $QA$.
Its total-variation data processing follows directly:
\begin{align*}
 \Delta(PA,QA)
 &=\frac12\sum_j\left|\sum_i\bigl((Pw)(i)-(Qw)(i)\bigr)
                                      P_{J\mid I}(j\mid i)\right|\\
 &\le\frac12\sum_i|(Pw)(i)-(Qw)(i)|\sum_jP_{J\mid I}(j\mid i)\\
 &=\Delta(Pw,Qw)\le2\alpha+2 d_Q.
\end{align*}
\end{proof}
The notation $Pw$ means pointwise multiplication on $\mathcal I$;
$PA$ is the explicitly defined sum over $i$ producing a distribution on
$\mathcal J$. An input slice fixes $i$ and varies the output index $j$. In general $\sum_iw(i)\ne1$. For example, $P=(1/2,1/2)$ and $w=(3/2,1/2)$ give $Pw=(3/4,1/4)$, although $w$ is not a distribution. The bound uses its mean under $P,Q$, not its unscaled sum. Figure~\ref{fig:classical-weight-bbj} shows both stages and their quantum-channel analogue. The scaled array $A$ itself is not assumed to be a classical channel.

\subsection{The complete fixed-message output scaling factors}
For a register value $\lambda=(y,y',\sflag,\overline z')$, define
\[
 \Sym(z,\lambda,c,c')=
 \begin{cases}
 \same,&\sflag=1,\ c'=c,\\
 f_z(c'),&\sflag=1,\ c'\ne c,\\
 f_{\overline z'}(c'),&\sflag=0.
 \end{cases}
\]
As before, assign $\perp$ on unsupported inconsistent labels to make the
rule total. In a real run, let $J_m=\Sym(Z,\Lambda,C,C')$.
Then $\Copy(m,J_m)$ is exactly the actual decoded message. To verify this, if $S=0$ the rule is the real tampered decoder;
if $S=1$ the key is unchanged; then either the whole codeword is unchanged
and gives $m$, or the rule uses precisely the unchanged-key decoder on
$C'$. These cases cover all runs and incur no exceptional event.

For $h_{y,y'}>0$, the following array combines the posterior conditional
distribution with message-selection scaling factors; it is not itself a conditional probability:
\begin{equation}
 A_m(j;z,\lambda)=\frac q{h_{y,y'}}
 \sum_{c\in A_{y,y'}}\ind[f_z(c)=m]
       \ind[j=\Sym(z,\lambda,c,g_2(y,c))].
 \label{eq:cl-kernel}
\end{equation}
Its input scaling factor is
\begin{equation}
 w_m(z,\lambda)=\frac q{h_{y,y'}}
                  \sum_{c\in A_{y,y'}}\ind[f_z(c)=m].
 \label{eq:cl-row-weight}
\end{equation}

\begin{lemma}[The two classical normalizations]\label[lemma]{lem:cl-normalization}
The distribution $\mu A_m$ is the exact real symbol distribution $P_{J_m}$, and
$\E_\mu w_m=\E_\sigma w_m=1$. In the ideal selected experiment, the
refined marginal $(\lambda,c,c')$ is exactly its message-free hybrid
marginal, independently of $m$.
\end{lemma}
\begin{proof}
Refining the actual coarse distribution uses \eqref{eq:cl-posterior} and
reproduces the actual uniform-ciphertext hybrid. Multiplying each refined register scaling factor by
$q\ind[f_z(c)=m]$ therefore conditions exactly on $M=m$, by
\eqref{eq:cl-real-law}. Its output symbol is $J_m$, and its mass is one.
Equivalently, for each $x,y$ the scaling factor sum is
$qN^{-1}T=1$, so the complete mass is $\sum_{x,y}p(x,y)=1$.

Let $\nu_{\Lambda CC'}$ be the hybrid's joint continuation distribution, obtained
by marginalizing $P^{\rm hyb}_{Z\Lambda CC'}$. Before computing the output
symbol, the ideal selected joint distribution is
\[
 \widetilde\sigma_m(z,\lambda,c,c')=
 q\,u_{\cK}(z)\ind[f_z(c)=m]\nu_{\Lambda CC'}(\lambda,c,c').
\]
For every $c$, the sum of its key-dependent factor over $z$ is one by
\cref{lem:cl-affine}. Therefore
$\sum_z\widetilde\sigma_m(z,\lambda,c,c')=
\nu_{\Lambda CC'}(\lambda,c,c')$ on the entire alphabet. Summing proves total mass one, and computing the
symbol preserves it. This proves both assertions for the ideal distribution.
\end{proof}

\subsection{Averaging posterior support sizes}
\begin{lemma}[Classical ideal deficit]\label[lemma]{lem:cl-deficit}
For $\sigma$ in \eqref{eq:cl-ideal},
\begin{equation}
  d_\sigma:=\E_\sigma(1-w_m)_+
             \le\frac12\sqrt{\frac{q2^d}{N}}
             =\frac12\,2^{-(r-d)/2}.
 \label{eq:cl-deficit}
\end{equation}
The estimate is the same for every $m$.
\end{lemma}
\begin{proof}
In the ideal joint distribution $\sigma_{Z\Lambda}$, let
$W=w_m(Z,\Lambda)$ and $H=h_{Y,Y'}$. The random variable $H$ is
positive on its support. Independence of $Z$ from $\Lambda$ and
\eqref{eq:cl-membership} give the conditional moments
\[
 \E_\sigma[W\mid\Lambda]=1,\qquad
 \E_\sigma[W^2\mid\Lambda]
   =\frac q{H}+\gamma\frac{H(H-1)}{H^2}
   =\gamma+\frac{q-\gamma}{H},
 \quad\gamma=\frac{N-q}{N-1}.
\]
Indeed, a posterior set of size $h$ has $h$ diagonal ordered pairs with
membership probability $1/q$ and $h(h-1)$ distinct pairs with probability
$T(T-1)/(N(N-1))$. Multiplication by $q^2/h^2$ gives this formula for
every value of the conditional distribution. Since $0\le\gamma\le1$,
\[
 \E_\sigma[(W-1)^2\mid\Lambda]
             =(\gamma-1)+(q-\gamma)/H\le q/H.
\]
The conditional mean of $W-1$ is zero. Its positive and negative
deviations therefore have equal conditional expectations. Conditional
Cauchy--Schwarz yields
\[
 \E_\sigma[(1-W)_+\mid\Lambda]
 =\tfrac12\E_\sigma[|1-W|\mid\Lambda]
 \le\tfrac12\sqrt{\E_\sigma[(1-W)^2\mid\Lambda]}
 \le\tfrac12\sqrt{q/H}.
\]
These are identities and bounds under the ideal joint distribution. We now average them.

The conditional deficit depends on $\Lambda$ only through $(Y,Y')$.
Taking expectation and using \eqref{eq:cl-yv} gives
\begin{align*}
  d_\sigma
 &\le\sum_y2^{-d}\sum_{y':h_{y,y'}>0}
       \frac{h_{y,y'}}N\cdot\frac12\sqrt{q/h_{y,y'}}\\
 &=\frac{\sqrt q}{2N}\sum_y2^{-d}
               \sum_{y':h_{y,y'}>0}\sqrt{h_{y,y'}}.
\end{align*}
For fixed $y$, there are at most $2^d$ nonempty sets and their sizes sum
to $N$. A second Cauchy--Schwarz inequality gives
\[
 \sum_{y':h_{y,y'}>0}\sqrt{h_{y,y'}}
 \le\sqrt{\bigl|\{y':h_{y,y'}>0\}\bigr|\sum_{y'}h_{y,y'}}
 \le\sqrt{2^dN}.
\]
Substitute and use $\sum_y2^{-d}=1$ to obtain
$ d_\sigma\le\sqrt{q2^d/N}/2$. The identity $N=2^{\ell+r}$,
$q=2^\ell$ proves the last equality. The counts were independent of the
message value, so the bound is uniform in $m$.
\end{proof}
This is why a classical proof needs no spectral smoothing. Even if some outputs are very unlikely, their small probability multiplies the
posterior estimate. The loss is controlled by the number of possible short
outputs and the average posterior size, not by the inverse of the least
positive probability. In the quantum proof, the continuation map must be a state-independent channel on an entire auxiliary-register space; the operator lower bound and smoothing provide a direct way to achieve that stronger requirement.

Applying \cref{lem:cl-weight} with the exact means in
\cref{lem:cl-normalization} yields, for $\Omega_m:=\sigma A_m$,
\begin{equation}
 \Delta(P_{J_m},\Omega_m)
       \le2\alpha+2 d_\sigma
       \le2\alpha+2^{-(r-d)/2}.
 \label{eq:cl-comparison}
\end{equation}
No factor $q\eta$ occurs. Only the message-selection scaling factors depend on $m$;
the two input distributions $\mu,\sigma$ and the bound are the same for all messages.

\subsection{The classical simulator and its exact ideal distribution}
The simulator for fixed deterministic $f,g$ samples $X,Y,C$ independently
uniform, computes $X'=f(X)$ and $(Y',C')=g(Y,C)$, and tests whether
$(X',Y')=(X,Y)$. If not, it outputs $f_{\nmext(X',Y')}(C')$. If yes and
$C'=C$, it outputs $\same$. Otherwise it outputs a fresh uniform
$\ell$-bit string. This is a message-free classical experiment.

Let $w_{\rm same},w_{\rm change}$ be the hybrid probabilities of
$S=1,C'=C$ and $S=1,C'\ne C$. Let $\gamma_u$ be the probability of
$S=0$ and tampered decoded output $u$. All are nonnegative and sum to one.
In the ideal selected experiment, the refined register distribution is unchanged
by \cref{lem:cl-normalization}. Conditional on a refined register with
$S=1$ and $c'\ne c$, the ideal key is uniform subject to $f_z(c)=m$, so
\cref{lem:cl-affine} gives distribution $p_m$ for $f_z(c')$. In the other two
cases, the output does not depend on the fresh key. Consequently
\begin{align*}
 \Omega_m&=w_{\rm same}\delta_{\same}
          +w_{\rm change}p_m+\sum_u\gamma_u\delta_u,\\
 D_{f,g}&=w_{\rm same}\delta_{\same}
          +w_{\rm change}U_\ell+\sum_u\gamma_u\delta_u.
\end{align*}
Their difference is exactly $w_{\rm change}(p_m-U_\ell)$, so
\begin{equation}
 \Delta(\Omega_m,D_{f,g})=w_{\rm change}\epsp\le\epsp.
 \label{eq:cl-ideal-sim}
\end{equation}
This ideal-key counting is not applied to the actual correlated key.
It is applied only after the comparison in \eqref{eq:cl-comparison}
has bounded the distance to the ideal output distribution.

\begin{theorem}[Classical non-malleable code]\label[theorem]{thm:classical}
Under \eqref{eq:cl-inner}, the non-malleable code described in~\eqref{eq:cl-code} has perfect
correctness and is a classical two-split NMC with
\begin{equation}
 \epsn^{\rm cl}
 \le2(\eta+2^{-s})+2^{-(r-d)/2}
                 +\frac{2^\ell-1}{2^\ell(2^s-1)}.
 \label{eq:classical-final}
\end{equation}
It permits arbitrary source-independent shared random tapes. Its rate is
$\ell/(n+d+s)$, with polynomial-time encoding and decoding.
\end{theorem}
\begin{proof}
For deterministic $f,g$, combine \eqref{eq:cl-comparison} and
\eqref{eq:cl-ideal-sim}, and apply stochastic data processing to the fixed
map $\Copy(m,\cdot)$. The exact real copy representation proves the
security bound for every $m$. The simulator was fixed using only $f,g$
and public parameters, so the required quantifier order holds.
Correctness and complexity were verified immediately after
\eqref{eq:cl-code}.

For randomized attacks, expose the entire pair of source-independent
random tapes $\Omega$ before the code is sampled. Conditional on any
value $\omega$, the two attacks are deterministic maps $f_\omega,g_\omega$.
The proof gives a simulator $D_\omega$ with the same parameter-only error
bound. Define $D$ by sampling $\Omega$ according to its original joint distribution
and then sampling $D_\Omega$. This distribution has no message input. If
$P_m^\omega$ is the real conditional output distribution, the triangle inequality
for finite probability sums gives
\begin{align*}
 &\Delta\left(\E_\Omega P_m^\Omega,
       \E_\Omega\Copy(m,D_\Omega)\right)\\
 &\qquad\le\E_\Omega\Delta(P_m^\Omega,\Copy(m,D_\Omega))\\
 &\qquad\le 2(\eta+2^{-s})+2^{-(r-d)/2}
                 +\frac{2^\ell-1}{2^\ell(2^s-1)}.
\end{align*}
Finite input and output alphabets give finitely many deterministic maps;
any distribution of tapes can be grouped by those map pairs, so the
argument does not require bounded tape descriptions. Shared coins thus
introduce no extra error.
\end{proof}

\subsection{The resulting constant rate, with all limits specified}
Let the classical $\nmext$ have
$d=\delta n+O(1)$ and $s_0=(1/2-\delta)n+O(1)$ output bits, as in
\eqref{eq:cl-inner-params}. Keep $2s=2\lfloor s_0/2\rfloor$ bits, set
$r=2d$, and let $\ell=s-r$. A fixed allowed $0<\delta<1/10$ gives
\[
 \ell=\frac{1-10\delta}{4}n+O(1),\quad r-d=d,\quad
 n+d+s=\frac{5+2\delta}{4}n+O(1).
\]
The first identity follows by halving the $\nmext$ output length and
subtracting $2d$; the last counts exactly the left source and right
source-plus-ciphertext. Dividing by $n$ proves limiting rate
$(1-10\delta)/(5+2\delta)$. Fixed smaller allowed values approach $1/5$.

For an $\nmext$ error $\eta\le2^{-c n^{a_{\rm err}}}$ with fixed $c,a_{\rm err}>0$, the first term in
\eqref{eq:classical-final} is stretched-exponentially small, the padding
term is $2^{-\delta n/2+O(1)}$, and the remaining seed and pair terms
are $2^{-\Omega(n)}$. Their finite sum is stretched-exponentially small
and in particular negligible. Every step of the outer argument is classical; BBJ is invoked only
to instantiate the explicitly stated inner attack.

\section{Elementary quantum facts used in the reduction}
\label{app:quantum-facts}
We give the finite-dimensional arguments behind \cref{lem:tools} and the
CP map formalism. The spectral theorem, basic properties of tensor products,
and ordinary Cauchy--Schwarz are the linear-algebra background. The
arguments in this appendix use only these linear-algebra facts.

\subsection{Kraus representation and probability operators}
Let $\Phi$ be a CP map from matrices on $A$ to matrices on $B$. For a basis
$\{\ket i\}$ of $A$, set $\ket\Omega=\sum_i\ket i\ket i$ and
\[
 J_\Phi=(\Id\otimes\Phi)(\ket\Omega\bra\Omega)
       =\sum_{i,j}\ket i\bra j\otimes\Phi(\ket i\bra j).
\]
Complete positivity makes $J_\Phi$ positive. Its spectral decomposition
can be written $J_\Phi=\sum_k\ket{v_k}\bra{v_k}$ by absorbing each
nonnegative eigenvalue into the corresponding vector. For each $k$ there
is a unique linear map $L_k:A\to B$ with
$\ket{v_k}=\sum_i\ket i\otimes L_k\ket i$. Comparing $(i,j)$ blocks
therefore gives
\[
 \Phi(\ket i\bra j)=\sum_kL_k\ket i\bra jL_k^\dagger.
\]
The matrix units span all matrices, so the Kraus formula holds for every
input. Conversely, a sum of such conjugations remains positive after
adjoining any reference, since it becomes conjugation by $I\otimes L_k$.
This proves the equivalence used in the paper.

Cyclicity of trace gives
$\Tr\Phi(T)=\Tr[(\sum_kL_k^\dagger L_k)T]$.
Equality to $\Tr T$ for all positive $T$ is equivalent to
$\sum_kL_k^\dagger L_k=I$, by testing rank-one projectors. The analogous
inequality is equivalent to $\sum_kL_k^\dagger L_k\preceq I$.
More generally,
$\Phi^*(B)=\sum_kL_k^\dagger B L_k$ satisfies the adjoint identity.
An outcome CP map's probability operator is $\Phi^*(I)$, while its
state update is the whole map $\Phi$.

To illustrate why the distinction matters, consider the qubit CP maps
\[
 \Phi_v(\rho)=\tfrac12\rho,
 \qquad
 \Psi_v(\rho)=\tfrac12\Tr(\rho)\ket v\bra v,
 \quad v\in\{0,1\}.
\]
Both families are instruments and both have probability operators $I/2$.
The first retains its input state on either outcome. The second discards
it and prepares the observed bit. Identical event probabilities do not
specify identical residual states, even without an entangled reference.
The exact map identity in \cref{lem:continuation} controls this additional
information.

\subsection{Variational formula and data processing}
For normalized $\rho,\sigma$, let $A=\rho-\sigma$. Its spectral positive
and negative parts obey $A=A_+-A_-$, $A_+A_-=0$, and
$\Tr A_+=\Tr A_- =\|A\|_1/2$, since $\Tr A=0$. For
$0\preceq B\preceq I$, positivity of traces gives
\[
 -\Tr A_-\le\Tr BA\le\Tr A_+.
\]
The projector onto the positive eigenspace attains the upper bound; the
negative eigenspace attains its absolute counterpart. This proves the
variational identity of \cref{lem:tools} and hence the bound on every
measurement event. The equality uses equal trace and is not silently
extended to unequal-trace subnormalized inputs.

If $\Phi$ is positive and trace-nonincreasing, then for Hermitian $A$,
\begin{align*}
 \|\Phi(A)\|_1
 &\le\|\Phi(A_+)\|_1+\|\Phi(A_-)\|_1\\
 &=\Tr\Phi(A_+)+\Tr\Phi(A_-)
 \le\Tr A_++\Tr A_-=\|A\|_1.
\end{align*}
Every CP trace-nonincreasing map is such a map. This proves data processing
on Hermitian differences. In particular, channels, including those that
discard quantum registers or compute classical functions, cannot increase
trace distance. Selecting
a classical event without renormalizing is a CP trace-nonincreasing map:
conjugate by its classical projector and optionally discard registers.
Normalized conditioning is a different, generally nonlinear operation.

For block-diagonal Hermitian matrices, the eigenvalues of a direct sum
are the union of its blocks' eigenvalues. Summing their absolute values
proves
\[
 \left\|\sum_x\ket x\bra x\otimes A_x\right\|_1
                 =\sum_x\|A_x\|_1.
\]
To tensor a difference with a fixed normalized state, diagonalize that
state as $\sum_jp_j\ket j\bra j$. The same identity gives
$\|A\otimes\theta\|_1=\sum_jp_j\|A\|_1=\|A\|_1$.
These observations prove all processing and classical-block assertions
used in \cref{lem:tools}.

\subsection{Purification and the state-scaled trace bound}
For unit vectors $u,v$, multiply $v$ by a harmless phase so their inner
product $c$ is real and nonnegative. If $c<1$, write
$v=cu+\sqrt{1-c^2}w$ with $w\perp u$. On the span of $u,w$, the difference
of their rank-one projectors has matrix
\[
 \begin{pmatrix}
 1-c^2&-c\sqrt{1-c^2}\\
 -c\sqrt{1-c^2}&-(1-c^2)
 \end{pmatrix}.
\]
Its eigenvalues are $\pm\sqrt{1-c^2}$, and its other eigenvalues are zero.
Thus the half trace norm is $\sqrt{1-c^2}$. If $c=1$, the projectors
are equal and the same formula gives zero. If $u,v$ purify states on a
smaller register, discard their reference and apply data processing. This
proves the pure-state assertion.

A purification of $\theta=\sum_i t_i\ket{e_i}\bra{e_i}$ is
$\Theta=\sum_i\sqrt{t_i}\ket{e_i}\ket i$.
For a positive $H$, applying $\sqrt H$ to its first register gives
squared norm $\Tr\theta H$ and overlap $\Tr\theta\sqrt H$ with
$\Theta$. These equalities follow by expanding the sums and using
orthogonality of $\ket i$. They justify the normalization and overlap
calculation in \cref{lem:weight-one}, without choosing a common eigenbasis
for $\theta$ and $H$.

For a Hermitian $A$, apply Hilbert--Schmidt Cauchy--Schwarz to
$\sqrt\theta$ and $|A|\sqrt\theta$. Their squared norms are
$\Tr\theta=1$ and $\Tr\theta A^2$, and their inner product is
$\Tr\theta|A|$. Hence
\[
 \Tr\theta|A|\le\sqrt{\Tr\theta A^2}.
\]
For a scalar nonnegative random variable $X$, ordinary Cauchy--Schwarz
applied to the vectors with entries $\sqrt{p_i}$ and $\sqrt{p_ix_i}$
gives $\E\sqrt X\le\sqrt{\E X}$. These are the two distinct steps used
in the moment-to-deficit estimate. The normalized trace prevents an
extraneous Hilbert-space dimension factor.

\subsection{No signalling versus selected outcomes}
Let $\omega_{AB}$ be bipartite and $\Phi:A\to A'$ be a channel with
Kraus operators $L_j$. For every test matrix $T_B$,
\[
 \Tr[(I_{A'}\otimes T_B)(\Phi\otimes\Id)(\omega)]
 =\sum_j\Tr[(L_j^\dagger L_j\otimes T_B)\omega]
 =\Tr[(I_A\otimes T_B)\omega].
\]
Equality for all $T_B$ proves that the reduced state on $B$ is unchanged.
This holds for each fixed classical value controlling the channel. It
requires the sum over all outcomes of a local instrument. On a single
selected outcome its map need not be $I$, and the conditional state on the other side can change. This is precisely why the classical posterior proof cannot simply be reused by replacing scalar conditional states with quantum states.

\clearpage
\section{Register diagrams and message conditioning}
\label{app:diagrams}
These original diagrams follow BBJ's use of named registers, local maps,
and intermediate states~\cite[Figures 6 and 11--14]{BBJ23}. They display
our extra message-selection step; they do not reproduce BBJ's figures.

\subsection{Notation and the registered tampering experiment}
\label{app:diagram-notation}
As in BBJ, $E_1,E_2$ denote auxiliary quantum registers and $E'_1,E'_2$
their outputs; $\cA_1,\cA_2$ are the channel forms of its local isometries
$U,V$. A register name specifies a system, whereas $\rho_{E_2}$ names its
state. The coarse map acts on $E_2$ itself, but changes its state.
The only deliberate differences in the payload notation are:
\begin{center}
\small
\begin{tabular}{@{}L{4cm}L{2cm}L{5cm}@{}}
\toprule
Object & BBJ & This paper\\\midrule
Extractor output / permutation key & $R$ & $Z=\nmext(X,Y)$\\[3pt]
Payload ciphertext & $Z$ & $C=\pi_Z(M\|R)$\\[3pt]
Fresh padding & --- & $R\leftarrow U_r$\\
\bottomrule
\end{tabular}
\end{center}
Thus $Z,C$ consistently denote the original key and ciphertext, the two
registers needed by the message test. Our $\mu,\sigma$ are defined here;
they are not identified with similarly named states in BBJ. A superscript
on $U_k^A$ specifies its register $A$, not a power.

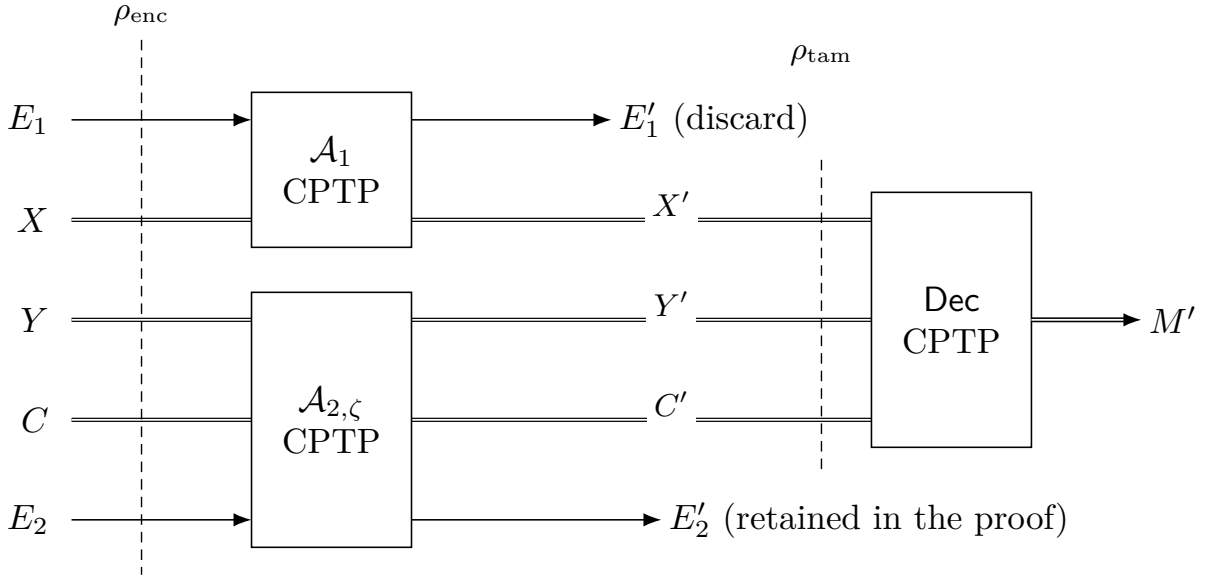
\begin{figure}[H]
\centering
\resizebox{\linewidth}{!}{%
\begin{tikzpicture}[x=1cm,y=1cm]
 \foreach \yy/\lab in {1/E_1,0/X,-1/Y,-2/C,-3/E_2} {
  \node[bbjstate,anchor=east] at (.25,\yy) {$\lab$};
 }
 \draw[bbjarr] (.4,1)--(2.2,1); \draw[bbjarr] (3.8,1)--(5.8,1);
 \node[bbjstate,anchor=west] at (5.8,1) {$E'_1$ (discard)};
 \draw[double,double distance=.5pt,line width=.35pt] (.4,0)--(2.2,0);
 \draw[double,double distance=.5pt,line width=.35pt] (3.8,0)--(8.4,0);
 \foreach \yy in {-1,-2} {
  \draw[double,double distance=.5pt,line width=.35pt] (.4,\yy)--(2.2,\yy);
  \draw[double,double distance=.5pt,line width=.35pt] (3.8,\yy)--(8.4,\yy);
 }
 \draw[bbjarr] (.4,-3)--(2.2,-3); \draw[bbjarr] (3.8,-3)--(6.3,-3);
 \node[bbjstate,anchor=west] at (6.3,-3) {$E'_2$ (retained in the proof)};
 \node[bbjbox,minimum width=1.6cm,minimum height=1.55cm] at (3,.5)
       {$\cA_1$\\CPTP};
 \node[bbjbox,minimum width=1.6cm,minimum height=2.55cm] at (3,-2)
       {$\cA_{2,\zeta}$\\CPTP};
 \node[bbjlab] at (6.4,.16) {$X'$};
 \node[bbjlab] at (6.4,-.84) {$Y'$};
 \node[bbjlab] at (6.4,-1.84) {$C'$};
 \node[bbjbox,minimum width=1.6cm,minimum height=2.55cm] at (9.2,-1)
       {$\Dec$\\CPTP};
 \draw[bbjarr,double,double distance=.5pt] (10,-1)--(11.1,-1);
 \node[bbjstate,anchor=west] at (11.1,-1) {$M'$};
 \draw[densely dashed] (1.1,1.8)--(1.1,-3.55);
 \draw[densely dashed] (7.9,.6)--(7.9,-2.55);
 \node[bbjlab] at (1.1,2.05) {$\rho_{\rm enc}$};
 \node[bbjlab] at (7.9,1.65) {$\rho_{\rm tam}$};
\end{tikzpicture}%
}
\caption{The two local quantum channels have no communication link. Double lines
carry classical registers; single lines carry quantum registers. The input on $E_1E_2$ is $\psi_{E_1E_2}$, independent of the encoding
on $X,(Y,C)$. Reserved classical copies are omitted; they are used only
by the analyst. Dashed cuts label states. The smoothing in $\cA_{2,\zeta}$
is a proof device, not part of the encoder.}
\label{fig:register-tampering}
\end{figure}

\Needspace{10\baselineskip}
\subsection{Why a uniform message gives the uniform-ciphertext experiment}
\label{app:uniform-diagram}
For this comparison only, take $M\leftarrow U_\ell$ as well as
$R\leftarrow U_r$, independently of $X,Y$. Then $M\|R$ is uniform on $s$
bits. Every raw seed specifies a bijection, so for every source pair
\[
 \Pr[C=c\mid X=x,Y=y]=N^{-1},\qquad M=f_{z(x,y)}(C).
\]
This is an exact distributional statement, even though the actual key
$z(x,y)$ need not be uniform. The two state-preparation channels below give
the same joint state, not merely the same ciphertext marginal.

\begin{figure}[H]
\centering
\resizebox{\linewidth}{!}{%
\begin{tikzpicture}[x=1cm,y=1cm]
 \node[bbjstate,text width=2.8cm] (a) at (1.4,0)
   {$U_\ell^M\otimes U_r^R$\\$\otimes U_n^X\otimes U_d^Y$};
 \node[bbjbox,text width=3.5cm,minimum height=1.55cm] (e) at (5.75,0)
   {$\mathcal E_{\rm rec}$ (CPTP)\\$Z=\nmext(X,Y)$\\$C=\pi_Z(M\|R)$};
 \node[bbjstate,text width=2.6cm] (o) at (10.15,0)
   {$\omega_{MZXYC}$};
 \draw[bbjarr] (a)--(e); \draw[bbjarr] (e)--(o);
 \node[bbjstate,text width=2.8cm] (b) at (1.4,-2.8)
   {$U_n^X\otimes U_d^Y$\\$\otimes U_s^C$};
 \node[bbjbox,text width=3.5cm,minimum height=1.55cm] (f) at (5.75,-2.8)
   {$\mathcal F_{\rm rec}$ (CPTP)\\$Z=\nmext(X,Y)$\\$M=f_Z(C)$};
 \node[bbjstate,text width=2.6cm] (p) at (10.15,-2.8)
   {$\omega_{MZXYC}$};
 \draw[bbjarr] (b)--(f); \draw[bbjarr] (f)--(p);
 \draw[bbjcmp] (o)--node[bbjlab,right] {exactly\\equal}(p);
\end{tikzpicture}%
}
\caption{Two equivalent preparations before tampering. The upper channel
retains $M,Z,X,Y,C$ and discards its padding register. The lower channel
computes the same $M$ from $Z,C$. Both preserve the indicated original
classical values. Appending the same state $\psi_{E_1E_2}$ and applying the attack of Figure~\ref{fig:register-tampering} preserves equality.
There is no approximation or BBJ invocation in this figure.}
\label{fig:uniform-bbj}
\end{figure}
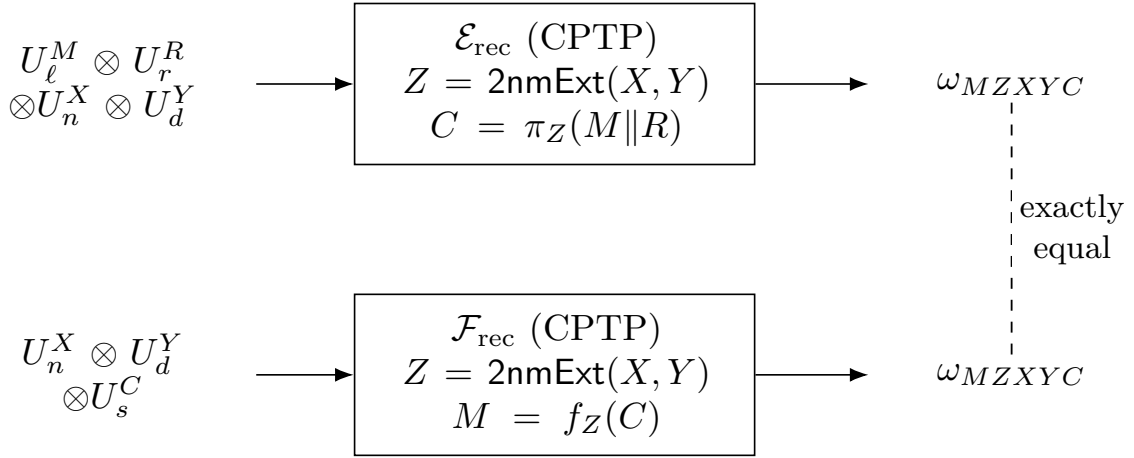

Indeed their common probability is, with $p(x,y)=2^{-n-d}$,
\[
 \Pr[M=m,X=x,Y=y,C=c]
   =\frac{p(x,y)}N\ind[f_{z(x,y)}(c)=m],
 \qquad Z=z(x,y).
\]
In the upper experiment exactly one padding string produces a given
$c$ with the indicated prefix, and it has probability $1/(qT)=1/N$.
In the lower experiment the predicate simply determines the computed $M$.
The conditional output state of the quantum attack is the same for each
classical input triple, so multiplying it by these identical probabilities
proves equality with $E'_2$ retained.

After the attack, let $\tau^+_{MZ\Lambda CC'E'_2}$ denote this
common state, with $M$ and all original values used by the test reserved.
It obeys $\Pr_{\tau^+}[M=m]=1/q$. It does \emph{not} generally factor as
$U_\ell^M$ times the other registers: the copied message is correlated with
$Z,C$. Selecting $M=m$ gives the fixed-message real experiment, by
\cref{lem:real-hybrid}. Replacing the contents of $M$ by $m$ without
selection would not give that experiment.

\Needspace{10\baselineskip}
\subsection{Exactly where BBJ is invoked}
\label{app:coarse-diagram}
Set $\rho_0=U_n^X\otimes U_d^Y\otimes\psi_{E_1E_2}$. Let
$\mathcal B$ be the following CPTP map: apply the original left operation
and the complete coarse right instrument, compute $Z,S,\overline Z'$ from
reserved classical source and output copies, retain
$\Lambda=(Y,Y',S,\overline Z')$ and $E_2$, and discard the other registers.
The computations of the flag and masked key are analyst operations, not
communication within the attack. Then $\mathcal B(\rho_0)=\mu_{Z\Lambda E_2}$.

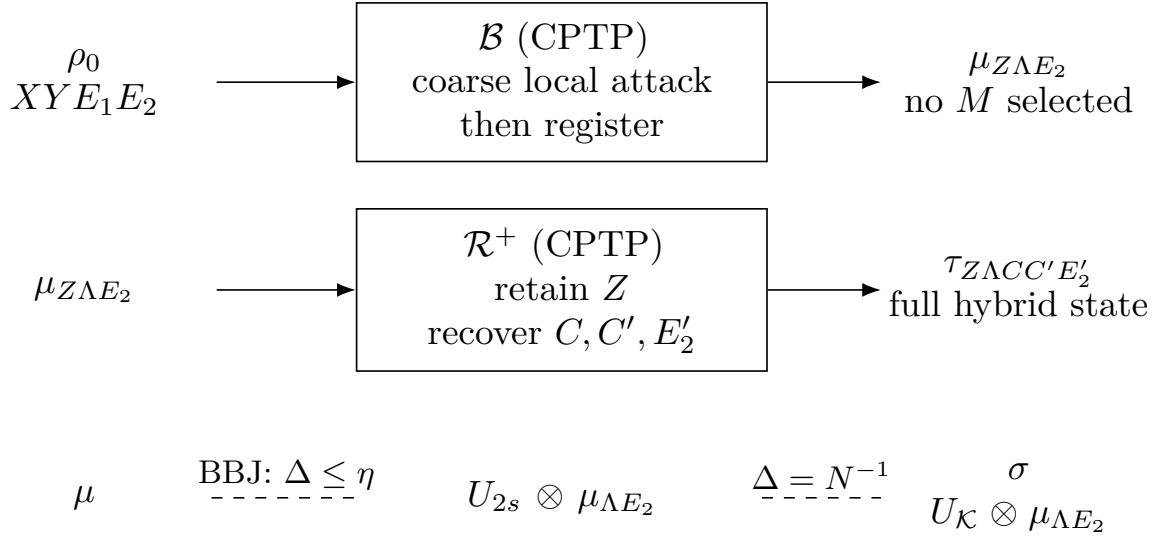
\begin{figure}[H]
\centering
\resizebox{\linewidth}{!}{%
\begin{tikzpicture}[x=1cm,y=1cm]
 \node[bbjstate,text width=2.4cm] (in) at (1.2,0) {$\rho_0$\\$XYE_1E_2$};
 \node[bbjbox,text width=3.6cm] (co) at (5.8,0)
  {$\mathcal B$ (CPTP)\\coarse local attack\\then register};
 \node[bbjstate,text width=2.5cm] (mu) at (10.2,0) {$\mu_{Z\Lambda E_2}$\\no $M$ selected};
 \draw[bbjarr] (in)--(co); \draw[bbjarr] (co)--(mu);
 \node[bbjstate,text width=2.4cm] (mu2) at (1.2,-2) {$\mu_{Z\Lambda E_2}$};
 \node[bbjbox,text width=3.6cm] (rr) at (5.8,-2)
  {$\mathcal R^+$ (CPTP)\\retain $Z$\\recover $C,C',E'_2$};
 \node[bbjstate,text width=2.5cm] (tau) at (10.2,-2) {$\tau_{Z\Lambda CC'E'_2}$\\full hybrid state};
 \draw[bbjarr] (mu2)--(rr); \draw[bbjarr] (rr)--(tau);
 \node[bbjstate,text width=2.3cm] (lo) at (1.2,-4) {$\mu$};
 \node[bbjstate,text width=3.7cm] (mid) at (5.8,-4) {$U_{2s}\otimes\mu_{\Lambda E_2}$};
 \node[bbjstate,text width=2.5cm] (hi) at (10.2,-4) {$\sigma$\\$U_{\cK}\otimes\mu_{\Lambda E_2}$};
 \draw[bbjcmp] (lo)--node[bbjlab,above] {BBJ: $\Delta\le\eta$}(mid);
 \draw[bbjcmp] (mid)--node[bbjlab,above] {$\Delta=N^{-1}$}(hi);
\end{tikzpicture}%
}
\caption{Solid arrows apply the indicated CPTP maps. The bottom dashed
connections are distance estimates, not physical operations that transform
one state into the other. BBJ is applied at $\mu$, \emph{before} continuation
and before any message selection. Only the original key is idealized; the
entire side state on $\Lambda E_2$ retains its actual marginal.}
\label{fig:coarse-bbj}
\end{figure}

The maps on the right register are exactly those of
\cref{lem:continuation}:
\[
 \mathcal C_{y'}^y(A)=\sqrt{G_{y'}^y}A\sqrt{G_{y'}^y},\qquad
 \mathcal R^{y,y'}_{c,c'}(A)=N^{-1}\mathcal V^{y,c}_{y',c'}
                  ((G_{y'}^y)^{-1/2}A(G_{y'}^y)^{-1/2}),
\]
\[
 \mathcal R^{y,y'}_{c,c'}\circ\mathcal C_{y'}^y
        =N^{-1}\mathcal V^{y,c}_{y',c'}.
\]
Each map is CP and trace-nonincreasing; the respective complete sums are
trace-preserving. The equality is on all input matrices, so it preserves
correlations with any reference. In $\mathcal R^+$, continuation is controlled
by $Y,Y'$ in $\Lambda$, and the original $Z$ is carried along unchanged.
On arbitrary input matrices its classical labels are first dephased, as in
\eqref{eq:Phi}. Consequently
\[
 \tau=\mathcal R^+(\mu),\qquad
 \tau_*:=\mathcal R^+(\sigma)=U_{\cK}\otimes\nu_{\Lambda CC'E'_2}.
\]
Neither $\tau$ nor $\tau_*$ has yet been conditioned on a message. The
conditional prefix test will need only $Z,C$, both now present. In particular,
there is no attempt to invoke BBJ on already conditioned source registers.

\Needspace{10\baselineskip}
\subsection{Computing a message, selecting it, and normalizing are different maps}
\label{app:selection-diagram}
Define $\mathcal F$ on the continued state by computing the classical
function $M=f_Z(C)$ into a fresh register while retaining all its inputs.
On computational basis states it is the isometry
$\ket{z,c}\mapsto\ket{f_z(c)}_M\ket{z,c}$, extended by the identity
on the other registers. In the earlier notation, $\mathcal F(\tau)=\tau^+$. For $Q_m=\ket m\bra m_M\otimes I$, define
\[
 \mathcal P_m(B)=\Tr_M(Q_m B Q_m).
\]
This is a CP trace-nonincreasing selection map. It retains the event's
probability in the trace. On a state whose classical $M$ was computed by
$\mathcal F$, its output is the same as projecting directly on $Z,C$:
\[
 \mathcal P_m\circ\mathcal F(A)=\Pi_m A\Pi_m,\qquad
 \Pi_m=\sum_{z,c:\,f_z(c)=m}\ket{z,c}\bra{z,c}\otimes I_{\Lambda C'E'_2}.
\]
Register labels identify the tensor factors; canonical register permutations
are implicit when the displayed order changes.
Only the original $Z,C$ are tested; the other registers have the identity
applied to them. On $\tau$ and $\tau_*$, the event has probability exactly
$1/q$. Thus the one fixed linear map
$\mathcal T_m(A)=q\Pi_m A\Pi_m$ produces a normalized output on both states.

\begin{figure}[H]
\centering
\resizebox{\linewidth}{!}{%
\begin{tikzpicture}[x=1cm,y=1cm]
 \node[bbjstate,text width=2.25cm] (a) at (1.15,0) {$\tau$ or $\tau_*$\\trace $1$};
 \node[bbjstate,text width=3.1cm] (b) at (5.8,0) {$\mathcal F(\tau)$ or $\mathcal F(\tau_*)$\\$\Pr[M=m]=1/q$};
 \node[bbjstate,text width=2.7cm] (c) at (10.2,0) {$\Pi_m\tau\Pi_m$\\or $\Pi_m\tau_*\Pi_m$\\trace $1/q$};
 \draw[bbjarr] (a)--node[bbjlab,above] {$\mathcal F$\\CPTP}(b);
 \draw[bbjarr] (b)--node[bbjlab,above] {$\mathcal P_m$\\CP--TNI} node[bbjlab,below] {$M=m$}(c);
 \node[bbjstate,text width=2.25cm] (d) at (1.15,-2.25)
   {selected state\\trace $1/q$};
 \node[bbjstate,text width=3.1cm] (e) at (5.8,-2.25)
   {$\mathcal T_m(\tau)$ or $\mathcal T_m(\tau_*)$\\trace $1$};
 \node[bbjstate,text width=2.7cm] (f) at (10.2,-2.25)
   {$P_{J_m}$ or $\Omega_m$\\trace $1$};
 \draw[bbjarr,rounded corners=2pt] (c.south) -- (10.2,-1.1)
    -- node[bbjlab,above] {same selected state} (1.15,-1.1) -- (d.north);
 \draw[bbjarr] (d)--node[bbjlab,above=5pt] {$A\mapsto qA$\\CP; not TNI}(e);
 \draw[bbjarr] (e)--node[bbjlab,above] {$\mathcal J$\\CPTP}(f);
\end{tikzpicture}%
}
\caption{The top row computes the original message and selects $M=m$; the
bottom row normalizes the selected state and computes the output symbol.
The connector carries the same selected state from the top row to the
bottom row; it performs no additional operation. CP-TNI means
completely positive and trace-nonincreasing. Multiplication by $q$ is an analytical normalization, not an outcome CP map on arbitrary inputs. The
same maps act on the actual and ideal states.}
\label{fig:conditioning-bbj}
\end{figure}
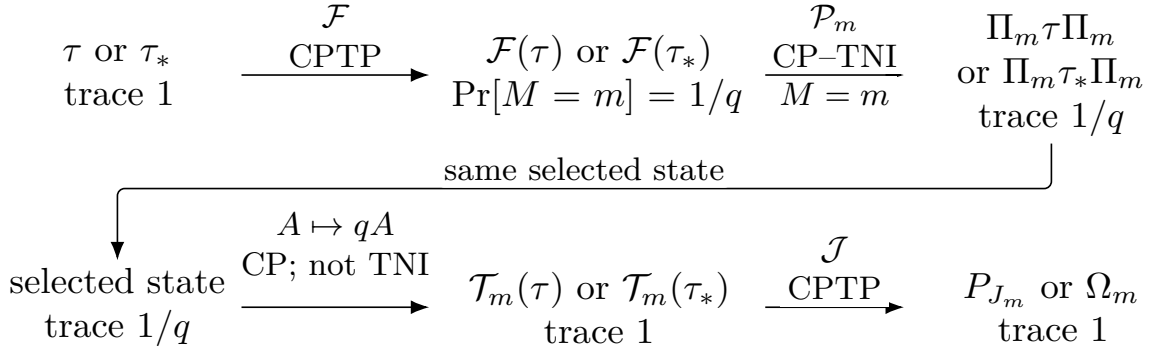

The event probability is the same for two different reasons. In the actual
state $\tau$, equal fiber size gives $T/N=1/q$ for every source pair. In the
ideal state $\tau_*=U_{\cK}\otimes\nu$, for every continued ciphertext $c$,
one-point uniformity gives $\Pr_{Z\leftarrow U_{\cK}}[f_Z(c)=m]=1/q$.
Neither statement asserts probability $1/q$ on every fixed key--transcript
pair. They are normalizations of the two complete states.

The map $\mathcal J$ computes \eqref{eq:j-symbol} and discards its other
registers. It still needs $Z$ when $S=1,C'\ne C$, which is why the key is
not discarded by the preceding maps. In particular,
\[
 \Phi_m=\mathcal J\circ\mathcal T_m\circ\mathcal R^+,
 \quad \Phi_m(\mu)=P_{J_m},\quad \Phi_m(\sigma)=\Omega_m.
\]

\Needspace{10\baselineskip}
\subsection{The two-state comparison through the same message-selection map}
\label{app:comparison-diagram}
A dashed comparison is not a sampling instruction. In particular, the top
estimate below is proved by BBJ and the seed restriction; it cannot be
interpreted as granting an attacker the ability to replace a derived key.
All classical labels in $\Lambda$ keep their actual joint marginal with
$E_2$. They are not recomputed from the independent ideal key.

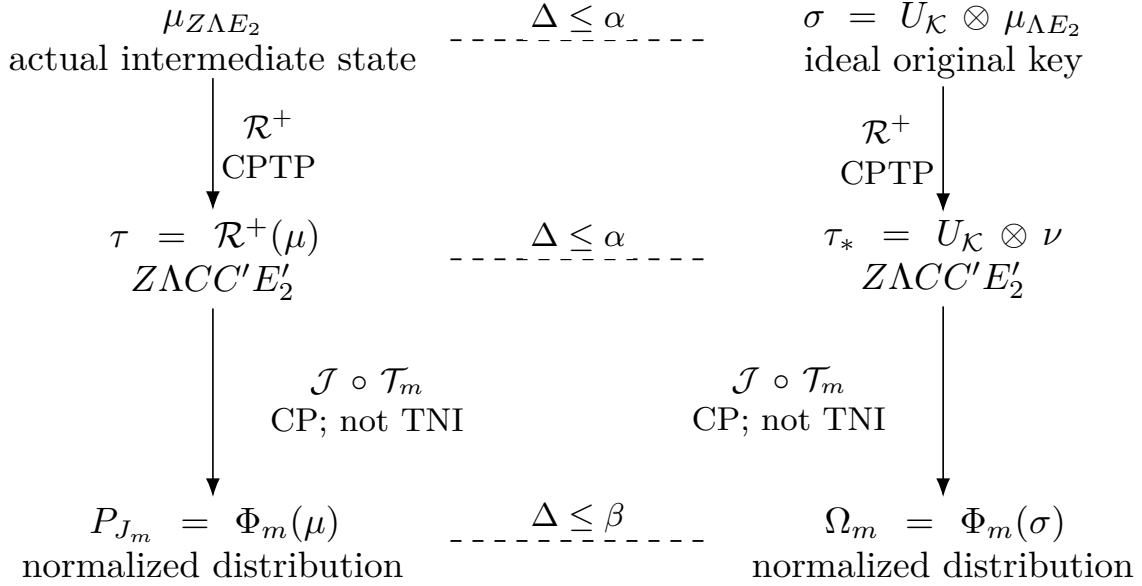
\begin{figure}[H]
\centering
\resizebox{\linewidth}{!}{%
\begin{tikzpicture}[x=1cm,y=1cm]
 \node[bbjstate,text width=4.4cm] (a) at (2.35,0) {$\mu_{Z\Lambda E_2}$\\actual intermediate state};
 \node[bbjstate,text width=4.4cm] (b) at (9.35,0) {$\sigma=U_{\cK}\otimes\mu_{\Lambda E_2}$\\ideal original key};
 \draw[bbjcmp] (a)--node[bbjlab,above] {$\Delta\le\alpha$}(b);
 \node[bbjstate,text width=4.4cm] (c) at (2.35,-2.1) {$\tau=\mathcal R^+(\mu)$\\$Z\Lambda CC'E'_2$};
 \node[bbjstate,text width=4.4cm] (d) at (9.35,-2.1) {$\tau_*=U_{\cK}\otimes\nu$\\$Z\Lambda CC'E'_2$};
 \draw[bbjarr] (a)--node[bbjlab,right] {$\mathcal R^+$\\CPTP}(c);
 \draw[bbjarr] (b)--node[bbjlab,left] {$\mathcal R^+$\\CPTP}(d);
 \draw[bbjcmp] (c)--node[bbjlab,above] {$\Delta\le\alpha$}(d);
 \node[bbjstate,text width=4.4cm] (e) at (2.35,-4.8) {$P_{J_m}=\Phi_m(\mu)$\\normalized distribution};
 \node[bbjstate,text width=4.4cm] (f) at (9.35,-4.8) {$\Omega_m=\Phi_m(\sigma)$\\normalized distribution};
 \draw[bbjarr] (c)--node[bbjlab,right,text width=2.8cm] {$\mathcal J\circ\mathcal T_m$\\CP; not TNI}(e);
 \draw[bbjarr] (d)--node[bbjlab,left,text width=2.8cm] {$\mathcal J\circ\mathcal T_m$\\CP; not TNI}(f);
 \draw[bbjcmp] (e)--node[bbjlab,above] {$\Delta\le\beta$}(f);
\end{tikzpicture}%
}
\caption{The bottom estimate is the mean-one scaling factor theorem, not ordinary
data processing of $\mathcal T_m$. The two solid paths compose to the same CP
map $\Phi_m$. Only this map depends on the requested $m$; the states
$\mu,\sigma,\nu$ and the continuation were defined without it. All distances
compare states on the same register space.}
\label{fig:comparison-bbj}
\end{figure}

The middle comparison is ordinary CPTP data processing. Applying that argument
to the subsequent rescaled selector would be incorrect: $\mathcal T_m$
can increase trace on other states. Instead the proof applies
\cref{thm:weighting} to the \emph{complete} $\Phi_m$, with
\[
 \alpha=\eta+N^{-1},\qquad
 H=\Phi_m^*(I),\qquad \Tr H\mu=\Tr H\sigma=1,
 \qquad d_\sigma\le\frac12\sqrt{2^{d-r}/\zeta}.
\]
It yields $\beta=\alpha+\sqrt{2\alpha}+2(2^{d-r}/\zeta)^{1/4}$.
The factor $q$ occurs inside the selector, but does not multiply $\eta$.

The final diagrammatic path can be read as the following distributional
comparisons, using the same deterministic $\Copy(m,\cdot)$ on each symbol:
\[
 \Dec(\cA(\Enc(m)))
 \ \stackrel{\zeta}{\approx}\ \Copy(m,P_{J_m})
 \ \stackrel{\beta}{\approx}\ \Copy(m,\Omega_m)
 \ \stackrel{\epsp}{\approx}\ \Copy(m,D_{\cA,\zeta}).
\]
Here $\stackrel{a}{\approx}$ means total variation distance at most $a$;
it is not a channel arrow. The distribution $D_{\cA,\zeta}$ is sampled
from the message-free hybrid, independently of $m$. The only residual
message dependence of $\Omega_m$ is the explicitly charged permutation
correction. Thus the diagrams preserve $\forall\cA\,\exists D\,\forall m$.

\subsection{The classical scaling factor lemma: one alphabet, then a new alphabet}
\label{app:classical-diagram}
In \cref{lem:cl-weight}, $P,Q$ are distributions on a finite set
$\mathcal I$. The function $w$ is a nonnegative \emph{scaling factor function} on
that set. The assumptions $\E_Pw=\E_Qw=1$ normalize the pointwise products
$Pw,Qw$; they do not state that $w$ is itself a probability distribution.
The array $A(j,i)$ additionally describes outputs indexed by another finite
set $\mathcal J$.

\begin{figure}[H]
\centering
\resizebox{\linewidth}{!}{%
\begin{tikzpicture}[x=1cm,y=1cm]
 \node[bbjstate,text width=2.3cm] (a) at (1.2,0) {$P$ on $\mathcal I$};
 \node[bbjstate,text width=2.6cm] (b) at (5.8,0) {$Pw$ on $\mathcal I$};
 \node[bbjstate,text width=2.5cm] (c) at (10.3,0) {$PA$ on $\mathcal J$};
 \node[bbjstate,text width=2.3cm] (d) at (1.2,-1.7) {$Q$ on $\mathcal I$};
 \node[bbjstate,text width=2.6cm] (e) at (5.8,-1.7) {$Qw$ on $\mathcal I$};
 \node[bbjstate,text width=2.5cm] (f) at (10.3,-1.7) {$QA$ on $\mathcal J$};
 \draw[bbjarr] (a)--node[bbjlab,above] {scaling factor $\times w$\\not a channel}(b);
 \draw[bbjarr] (d)--node[bbjlab,below] {same scaling factor $w$\\both means one}(e);
 \draw[bbjarr] (b)--node[bbjlab,above] {$P_{J\mid I}$\\classical channel}(c);
 \draw[bbjarr] (e)--node[bbjlab,below] {the same $P_{J\mid I}$\\classical channel}(f);
 \draw[bbjcmp] (a)--node[bbjlab] {$\Delta\le\alpha$}(d);
 \draw[bbjcmp] (b)--node[bbjlab] {$\Delta\le2\alpha+2 d_Q$}(e);
 \draw[bbjcmp] (c)--node[bbjlab] {no larger\\distance}(f);
\end{tikzpicture}%
}
\caption{Pointwise rescaling stays on the input alphabet. The second
operation samples a new value from a conditional distribution, so it may
change the alphabet. The last pair of arrows contracts total variation;
the first pair need not. The proof bounds the first stage by bounded deficits.}
\label{fig:classical-weight-bbj}
\end{figure}
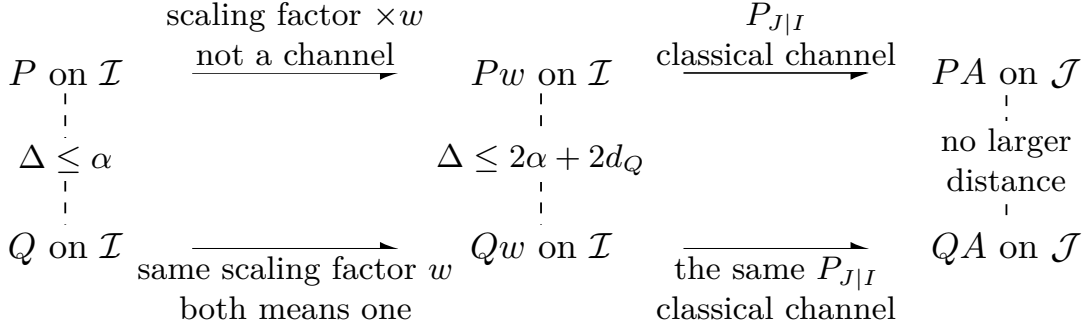

For $w(i)>0$, $P_{J\mid I}(j\mid i)=A(j,i)/w(i)$. Zero-scaling factor input slices can
be completed arbitrarily because both rescaled inputs assign them mass
zero. The output formula is
\[
 (PA)(j)=\sum_{i\in\mathcal I}P(i)w(i)P_{J\mid I}(j\mid i),
\]
not the pointwise product $P(j)A(j)$. The data processing in
\cref{lem:cl-weight} is precisely the inequality
$\Delta(PA,QA)\le\Delta(Pw,Qw)$.

\heading{Its exact CP-map analogue.}
Identify a distribution on $\mathcal I$ with a diagonal density matrix.
Put $D_w=\sum_iw(i)\ket i\bra i$. The map
$\mathcal W(B)=\sqrt{D_w}B\sqrt{D_w}$ is CP and sends
$\operatorname{diag}(P)$ to $\operatorname{diag}(Pw)$.
The classical channel is the CPTP map
\[
 \mathcal K(B)=\sum_{i,j}P_{J\mid I}(j\mid i)
       \langle i|B|i\rangle\ket j\bra j.
\]
Its Kraus operators $K_{ji}=\sqrt{P_{J\mid I}(j\mid i)}\ket j\bra i$
satisfy $\sum_{i,j}K_{ji}^\dagger K_{ji}=I$. Therefore
$(\mathcal K\circ\mathcal W)(\operatorname{diag}P)
=\operatorname{diag}(PA)$. This is the classical version of factoring a
CP map through its output-scaling factor operator. It is not a claim that
pointwise scaling itself is trace-preserving.

For example, $P=(1/3,1/3,1/3)$ and $w=(0,1,2)$ give
$Pw=(0,1/3,2/3)$. Taking the three conditional output distributions to
be $(1,0)$, $(1/2,1/2)$, and $(0,1)$ gives $PA=(1/6,5/6)$.
The input alphabet has three values, the output alphabet two;
$\sum_iw(i)=3$ has no normalization role.

\end{document}